\documentclass[11pt, smallextended,envcountsect,envcountsame]{amsart} 

\usepackage{graphicx}
\usepackage{enumerate}\usepackage{anysize}
\usepackage{fancyhdr}
\usepackage{amsmath}
\usepackage{amssymb}
\usepackage{amsxtra}
\usepackage{amsfonts}
\usepackage{bbold}
\usepackage{subfigure}
\usepackage{longtable}
\usepackage{caption}
\usepackage{tabularx}
\usepackage{xcolor}
\usepackage{tikz}
\usepackage{comment}

\numberwithin{equation}{section}
\numberwithin{table}{section}
\newtheorem{Thm}{Theorem}[section]
\newtheorem{lem}[Thm]{Lemma}

\newtheorem{prop}[Thm]{Proposition}

\theoremstyle{definition}
\newtheorem{defn}{Definition}[section]

\newcommand{\eee}{{\rm e}}

\newcommand{\one}[1]{\mathbb{1}_{\left\{#1\right\}}}
\newcommand{\card}[1]{\left| #1 \right|}
\newcommand{\bbd}[1]{\boldsymbol{#1}}

\newcommand{\dprime}{\prime\prime}

\def\tende#1{\vtop{\ialign{##\crcr\rightarrowfill\crcr
			\noalign{\kern-1pt\nointerlineskip}
			\hskip3.pt${\scriptstyle #1}$\hskip3.pt\crcr}}}

\begin{document}
	\title[New Results on the domain of analyticity   of the free energy for the Ising model]{New Results on the domain of analyticity of  the free energy for the Ising model}
	\author{Nguyen Tong Xuan}
	\address{Quy Nhon University, 170 An Duong Vuong street, Quy Nhon Nam ward, Gia Lai province, Vietnam}
	\email{nguyentongxuan@qnu.edu.vn}
	
	\author{Nguyen Dang Thien Thu} 
	\address{Quy Nhon University, 170 An Duong Vuong street, Quy Nhon Nam ward, Gia Lai province, Vietnam}
	\email{nguyendangthienthu@qnu.edu.vn }
	
	\date{Received: date / Accepted: date}
	
	\thanks{}
	\maketitle
	\begin{abstract}
		We investigate the analyticity of the free energy of the Ising model in the presence of a non-zero external magnetic field, at high temperature, and at low temperature. Using the Fernandez--Procacci convergence criterion for cluster expansions, together with generating-function techniques and graph-theoretical methods, we derive improved convergence conditions in all three regimes. In particular, the generating-function approach yields sharper estimates for polymers and contours in the strong-field and low-temperature regimes, while a new high-temperature expansion based on Veblen's theorem provides a substantially larger analyticity region than the classical results in the literature.
	\end{abstract}
	
	\section{INTRODUCTION} 
	
	The Ising model is one of the most fundamental and influential models in statistical mechanics, originally introduced to describe ferromagnetic phenomena. In this model, each site of a crystal lattice is assigned a spin variable taking one of two possible values, $+1$ or $-1$, representing the two possible orientations of a magnetic dipole moment. The spins interact with their nearest neighbors through pair interactions and may also be subject to an external magnetic field. Despite its apparent simplicity, the Ising model has become a paradigmatic framework for studying a wide range of physical phenomena, including phase transitions, spontaneous symmetry breaking, spontaneous magnetization, and numerous mathematical questions concerning the existence and analytic properties of thermodynamic quantities. A comprehensive account of the Ising model and its mathematical foundations can be found in the monograph by Friedli and Velenik \cite{FV17}.
	
	Among the thermodynamic quantities associated with the Ising model, the free energy plays a central role. It is defined as the logarithm of the partition function and encodes essentially all equilibrium thermodynamic information of the system. The analytic properties of the free energy are intimately related to the occurrence of phase transitions, since singularities of the free energy characterize qualitative changes in the macroscopic behavior of the system. Beyond the Ising model, the study of the analyticity domain of the free energy has also attracted considerable attention across numerous statistical mechanics model such as the Potts model, the Blume--Capel model, and the Curie--Weiss model.
	
	One of the most powerful tools for studying the analyticity of the free energy is the cluster expansion. The basic idea is to represent the logarithm of the partition function as a convergent power series with respect to suitable auxiliary parameters, commonly referred to as fugacities or activities. The cluscter expansion was originally developed in the context of the virial expansion for gases and liquids and has since evolved into a fundamental technique with numerous applications in statistical mechanics, probability theory, and combinatorics. An extensive overview of cluster expansion and its applications can be found in \cite{FV17,JPS13,NPS12,NF20}.
	
	A fundamental issue in the theory of cluster expansions is the determination of convergence conditions, since the convergence of the expansion immediately implies the analyticity of the free energy. Since the late 1960s, several convergence criteria have been developed using different approaches, including the Kirkwood--Salzburg equations \cite{GK71}, tree-graph inequalities \cite{Bry84}, Dobrushin's inductive method \cite{Dob96}, and the partition scheme introduced by Fernandez and Procacci \cite{FP07}. Among these approaches, the Fernandez--Procacci criterion currently provides one of the strongest known convergence conditions and has been successfully refined and applied to a broad class of polymer models \cite{BFP10,NF20}.
	
	For the Ising model, the polymer representation of the partition function depends significantly on the parameter regime under consideration. In the presence of a strong external magnetic field $(h\ne 0)$, spin configurations opposing the external field appear only as rare excitations and naturally form a dilute polymer gas. In the high-temperature regime, the interaction between neighboring spins is sufficiently weak so that the partition function admits a polymer representation based on the classical high-temperature expansion. In contrast, in the low-temperature regime, the ground states dominate the Gibbs measure, and the relevant excitations are described by Peierls contours, leading to a contour representation of the partition function. Although all three representations rely on the same general principle of cluster expansion, they give rise to different polymer systems and consequently require different convergence analyses.
	
	In this paper, we investigate the analyticity domain of the free energy of the Ising model in the three parameter regimes described above by combining the cluster expansion with the convergence criterion of Fernandez and Procacci. More precisely, for the non-zero external field and the low-temperature regime, besides applying the Fernandez--Procacci convergence criterion, we incorporate the generating-function technique introduced by Balister and Bollobás in \cite{BB07}. Originally, this technique was developed to establish upper and lower bounds for the number of bounded regions arising in graph arrangements. In the present work, we adapt this method to obtain sharper estimates for the number of polymers (or contours) of a given size. Combining these refined combinatorial estimates with the Fernandez--Procacci criterion leads to substantially improved convergence conditions for the cluster expansion and, consequently, to larger analyticity domains of the free energy in both regimes. 
	
	For the high-temperature regime, rather than relying on the standard high-temperature polymer representation, we develop a new cluster expansion based on tools from graph theory, with Veblen's theorem playing a key role in the characterization of the Eulerian subgraphs arising in the expansion. This new representation allows the Fernandez--Procacci convergence criterion to be applied in a more effective manner. As a consequence, we obtain an analyticity region for the free energy that is significantly larger than those previously available in the literature, including the classical results of Friedli and Velenik \cite{FV17} and Simon \cite{Simon93}. Our results demonstrate that combining modern combinatorial enumeration techniques and graph-theoretical methods with the cluster expansion provides a unified framework for analyzing the Ising model in different parameter regimes while yielding substantially improved analyticity domains for the free energy.
	\section{ISING MODEL AND MAIN RESULTS}\label{sec.Xuan.1}
	\subsection{Introduction to Ising model}
	Let us begin with recalling the definition of distances, 
	\[d(i,j)\;:=\;\|i-j\|_\infty\;=\;\max_{1\le k\le d} \card{i_k-j_k}\]
	for $\;i,j\in \Lambda$, and 
	\[d(S,S')\;:=\;\inf\{d(\kappa,\ell):\kappa\in S,\ell\in S'\}\]
	for   $S,S'\subset \mathbb Z^d$.

	We consider the set $\Omega=\{-1,1\}^{\mathbb{Z}^d}$ with $d\ge 2$.  Configurations are denoted by $\boldsymbol{\sigma}=(\sigma_x)_{x\in\mathbb Z^d}$. Let us consider a finite set $\Lambda\subset\mathbb{Z}^d$, configurations $\boldsymbol{\sigma}_{\Lambda}\in \Omega_{\Lambda}:=\{-1, 1\}^{\Lambda}$ and the Hamiltonian with free boundary condition
	\begin{equation}\label{eq:int.free1}\mathrm{H}^{\varnothing}_{\Lambda;\beta, h}(\boldsymbol{\sigma}_{\Lambda}):=-\beta\sum_{\{i,j\}\in\Lambda}f(i,j)\sigma_i\sigma_j-h\sum_{i\in\Lambda}\sigma_i, 
	\end{equation}
	where $\beta\in\mathbb{R}_{\ge0}$ is the inverse temperature, $h\in\mathbb{R}$ is the external field, and the interaction $f(\cdot, \cdot)$ is defined as 
	\begin{equation}\label{eq:int.free2}
		f(i,j)=\left\{\begin{array}{cl} 1 & \quad\mbox{ if }\|i-j\|_\infty=1\\ 0 & \quad\mbox{ otherwise }
		\end{array}\right..
	\end{equation}
	For each configuration $\boldsymbol{\sigma}_\Lambda\boldsymbol{\omega}_{\Lambda^c}\in\Omega$, Hamiltonian is defined as
	\begin{align}\label{eq:int.is1}\mathrm{H}^{\boldsymbol{\omega}}_{\Lambda; \beta, h}(\boldsymbol{\sigma}_{\Lambda}\boldsymbol{\omega}_{\Lambda^c})&=\mathrm{H}^{\varnothing}_{\Lambda;\beta, h}(\boldsymbol{\sigma}_{\Lambda})-\beta\sum_{i \in\Lambda\atop j\in \Lambda^c}f(i,j)\sigma_i\omega_j,
	\end{align}
	where the interaction $f(\cdot, \cdot)$ is defined in \eqref{eq:int.free2}, a configuration $\boldsymbol{\sigma}_\Lambda\boldsymbol{\omega}_{\Lambda^c}\in\Omega$ includes two parts $\boldsymbol{\sigma}_\Lambda\in =\Omega_{\Lambda}$, and $ \boldsymbol{\omega}_{\Lambda^c}\in \{-1, 1\}^{\Lambda^c}$ which is usually called a boundary of the systems, or configurations are frozen outside of the finite set $\Lambda$, and the term,
	\[\beta\sum_{i \in\Lambda\atop j\in \Lambda^c}f(i,j)\sigma_i\omega_j,\]
	refers to \emph{the interaction between the internal and external components of the system}.   
	
	The \emph{partition function} with free boundary condition in $\Lambda$ is  
	\begin{equation}\label{eq:int.is2}Z^{\varnothing}_{\Lambda}(\beta,h)\;=\;\sum_{\boldsymbol{\sigma}_{\Lambda}\in\Omega_{\Lambda}}\exp\left(-\mathrm{H}^{\varnothing}_{\Lambda;\beta,h}(\boldsymbol{\sigma}_{\Lambda})\right),\end{equation}
	the (finite-volume)  free energy function (pressure function) with free boundary condition is
	\begin{equation}\label{eq:int.is3}P^{\varnothing}_{\Lambda}(\beta,h)\;=\;\frac{1}{\card\Lambda}\log Z^{\varnothing}_{\Lambda}(\beta,h).\end{equation}
	The \emph{partition function} with $\boldsymbol{\omega}-$boundary condition in $\Lambda$ is  
	\begin{equation}\label{eq:int.is2}Z^{\boldsymbol{\omega}}_{\Lambda}(\beta,h)\;=\;\sum_{\boldsymbol{\sigma}_{\Lambda}\in\Omega_{\Lambda}}\exp\left(-\mathrm{H}^{\boldsymbol{\omega}}_{\Lambda;\beta,h}(\boldsymbol{\sigma}_{\Lambda}\boldsymbol{\omega}_{\Lambda^c})\right),\end{equation}
	the (finite-volume) free energy function with $\boldsymbol{\omega}-$boundary condition in $\Lambda$  is
	\begin{equation}\label{eq:int.is3}P^{\boldsymbol{\omega}}_{\Lambda}(\beta,h)\;=\;\frac{1}{\card\Lambda}\log Z^{\boldsymbol{\omega}}_{\Lambda}(\beta,h).\end{equation}
	The thermodynamic free energy function $p^{\#}$ is obtained through the thermodynamic limit
	\begin{equation}\label{eq:int.is4}p^{\#}(\beta,h)\;=\;\lim\limits_{\Lambda\uparrow\mathbb{Z}^d}P^{\#}_{\Lambda}(\beta,h)\end{equation}
	in Fisher sense, where $\#:= \varnothing \mbox{ or }\boldsymbol{\omega}$. 
	
	\subsection{Ising model in strong field ($h\ne 0$)}\label{subsection.Ising.strong.Xuan1}
	We utilize the fact that thermodynamic pressure is independent of boundary conditions (for reference, please take a look at Theorem 3.8 \cite{FV17}) and, for the sake of algebraic convenience, we will focus in this section on ``plus'' boundary conditions: $\omega_i=1$ for all $i\not\in\Lambda$.  The interaction between the inside and outside of the system can be described as the following term
	\[
	\beta\sum_{\scriptstyle i\in\Lambda, \,j\not\in\Lambda\atop\scriptstyle \|i-j\|_\infty=1}\sigma_i\;.
	\]
	To get the presentation of partition function, we add and subtract $1$ to each term in this Hamiltonian and 
	for each $\boldsymbol{\sigma}_{\Lambda}\in\Omega_{\Lambda}$, let us introduce the set
	\begin{equation}\label{eq:Is.mag3}
		\Lambda^-(\boldsymbol{\sigma}_{\Lambda})\;=\;\{i\in\Lambda:\;\sigma_i=-1\}\;.
	\end{equation}
	We obtain 
	\begin{align}\label{eq:Is.mag4}
		\mathrm{H}^{+}_{\Lambda;\beta,h}(\boldsymbol{\sigma}_{\Lambda})=&-\beta\card{\mathcal{E}_{\Lambda}}-h\card{\Lambda}+2\beta\card{\partial\Lambda^-(\boldsymbol{\sigma}_{\Lambda})} +2h\card{\Lambda^-(\boldsymbol{\sigma}_{\Lambda})}
	\end{align}
	where  
	\begin{eqnarray*}\partial\Lambda^-(\boldsymbol{\sigma}_\Lambda)&=&\Big\{\{i,j\}:i\in \Lambda^-(\boldsymbol{\sigma}_{\Lambda}), j\notin\Lambda^-(\boldsymbol{\sigma}_{\Lambda}),\, \|i-j\|_\infty=1\Big\}
	\end{eqnarray*}
	and 
	\begin{equation}\label{eq.Xuan.edge.1}
		\mathcal{E}_{\Lambda}=\big\{\{i,j\}\subset \mathbb{Z}^d:\{i,j\}\cap\Lambda\ne\varnothing,\,\|i-j\|_\infty=1\big\}\;.
	\end{equation}
	Each $\boldsymbol{\sigma}_{\Lambda}$ corresponds one-to-one to a term of deviation from the ground state $\Lambda^-(\boldsymbol{\sigma}_{\Lambda})$ (the configuration with minimal energy), which is the ``all $+1$'' configuration. As a consequence, the partition function can be expressed in terms of deviations from the ground state:
	\begin{equation}\label{eq:Is.mag5}
		Z^{+}_{\Lambda}(\beta,h)\;=\;\exp(\beta |\mathcal E_{\Lambda}|+h|\Lambda|)\Xi^{\mathrm{LF}}_{\Lambda}(\beta,h),
	\end{equation}
	where  \emph{the large field polymers partition function} $\Xi^{\rm LF}_{\Lambda}(\beta, h)$ is given as
	\[\Xi^{\mathrm{LF}}_{\Lambda}(\beta,h):=\sum_{\Lambda^-\subset \Lambda}\exp\left(-2\beta|\partial\Lambda^-|-2h|\Lambda^-|\right).\]
	From the definition of the distance, let us declare that two vertices $i,j\in\Lambda^-$ are \emph{ connected} if and only if $d(i,j)\le1$, and we can  decompose $\Lambda^-$ into maximally connected components (For example, see Figure \ref{fig:Ising_models1}),
	\[\Lambda^-=S_1\cup...\cup S_n\]
	with  $d(S_\ell,S_k)>1$ for $ \ell\ne k$. 
	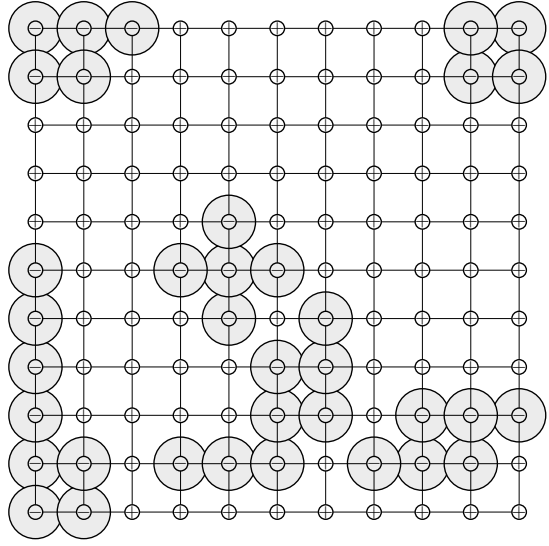
\begin{figure}[htb]
		\begin{center}
			\resizebox{0.45\textwidth}{!}{
				\begin{tikzpicture}
					
					\def\rows{10}
					\def\cols{10}
					

					\foreach \x/\y in {0/1, 0/0, 0/2, 1/0, 0/3, 0/4, 0/5,
						0/10, 1/10, 0/9, 1/9, 8/1,
						1/1,  3/1, 4/1, 5/1, 5/2, 5/3, 6/2,6/3, 6/4, 4/4, 4/5,
						7/1,  8/2, 9/1, 10/2, 3/5, 4/6, 5/5, 10/10, 9/9,10/9, 9/10, 9/2, 2/10
					} {
						\fill[gray!15] (\x,\y) circle(0.55);
						\draw[thick] (\x,\y) circle(0.55);
						\node at (\x,\y) {\textbf{$-$}};
					}
					\foreach \x/\y in {  0/6, 0/7, 0/8,
						1/2, 1/3, 1/4, 1/5, 1/6, 1/7, 1/8,
						2/0, 2/1, 2/2 ,2/3, 2/4, 2/5, 2/6, 2/7, 2/8, 2/9, 
						3/0, 3/2, 3/3, 3/6, 3/4,  3/7, 3/8, 3/9, 3/10, 4/0, 4/2, 4/3, 4/7, 4/8, 4/9, 4/10, 5/0, 5/4, 5/6, 5/8, 5/9, 5/10, 5/7, 6/0, 6/1, 6/5,6/6, 6/7, 6/8, 6/9, 6/10,
						7/0, 7/2, 7/4, 7/5, 7/7, 7/8, 7/9, 7/10,7/3, 7/6,
						8/0,  8/4, 8/5, 8/7, 8/8, 8/9, 8/10,8/3, 8/6,
						9/0,  9/4, 9/5, 9/7, 9/8, 9/3, 9/6,
						10/0, 10/1,  10/4, 10/5, 10/7, 10/8, 10/3, 10/6 
					} {
						\node at (\x,\y) {\textbf{$+$}};
					}
					\foreach \x in {0,...,\cols} {
						\foreach \y in {0,...,\rows} {
							\ifnum\x<10
							\draw (\x+0.15,\y) -- (\x+0.85,\y);
							\fi
							\ifnum\y<10
							\draw (\x,\y+0.15) -- (\x,\y+0.85);
							\fi
						}
					}
					\foreach \x in {0,...,\cols} {
						\foreach \y in {0,...,\rows} {
							\draw[thick] (\x, \y) circle (0.15);
						}
					}
					
					
			\end{tikzpicture}}
		\end{center}
		\caption{ A configuration of the Ising model. Each connected component of the shaded area delimits one of the polymers $S_1,\ldots, S_6$.}
		\label{fig:Ising_models1}
	\end{figure}
	Before giving an alternative expression of large field polymers partition function, let us introduce the definitions of compatible and incompatible objects as follow.
	\begin{defn}
		Let us define $S, S^{\prime}$ to be  compatible, and denote $S\sim S^{\prime}$, if $d(S,S^{\prime})\ge 2$. Otherwise $S$ and $S^{\prime}$ are incompatible and we denote $S\nsim S^{\prime}$.
	\end{defn} 
	\noindent
	Denote
	\begin{equation}\label{eq.Ising.model.Xuan.1}
		\zeta(S,S')=\left\{\begin{array}{cl} 1 & \quad \mbox{ if }\quad S\sim S'\\ 0 & \quad \mbox{ if }\quad S\nsim S'\end{array}\right..
	\end{equation}
	Since $|\partial_e\Lambda^-|=\sum_{i=1}^n|\partial_e S_i|$ and $|\Lambda^-|=\sum_{i=1}^n|S_i|$, then the expression of  large field polymers partition function can be rewritten in the following form
	\begin{align}\label{eq:Is.mag7}
		\Xi_{\Lambda}^{\mathrm{LF}}(\beta,h)=1+\sum_{n\ge 1}\frac{1}{n!}\sum_{(S_1,\ldots, S_n)\in\mathcal{P}_{\Lambda}^n}\prod_{1\le i<j\le n}\zeta(S_i, S_j)\prod_{i=1}^nw_{\beta,h}(S_i),
	\end{align}
	with 
	\[\mathcal P_{\Lambda}\;=\;\left\{S\subset \Lambda:\; S\; \mbox{is non-empty and connected of }\;\Lambda\right\}\]
	and 
	\begin{equation}\label{eq:Is.mag.add8}
		w_{\beta,h}(S_i)\;=\;\exp(-2\beta |\partial S_i|-2h|S_i|).
	\end{equation} 
	Let $\mathcal{C}[n]$ be the set of connected graph on $n$ vertices. Let $E(G)$ be the set of edges in the graph $G$. We set  
	\[\mathcal{P}_{\mathbb Z^d}=\left\{S\subset \mathbb Z^d:\; S\; \mbox{is non-empty and connected of }\;\mathbb Z^d\right\}.\]
	Using the cluster expansion theory (see reference \cite{BFP10, FP07, NF20} for more detail), we then can express the pressure with $+1$-boundary condition in $\Lambda$ as following form:
	\begin{equation}\label{eq:Is.mag.strong.16}
		P^{+}_{\Lambda}(\beta, h)=\beta\frac{|\mathcal{E}_{\Lambda}|}{|\Lambda|}+h+\frac{1}{|\Lambda|}\log \Xi^{\rm LF}_{\Lambda}(\beta,h),
	\end{equation}
	where
	\begin{align}\label{eq:Is.X.mag.stronggg16}
		\log \Xi^{\rm LF}_{\Lambda}(\beta,h)=\sum_{n=1}^{\infty}\frac{1}{n!}\sum_{(S_1\ldots S_n)\in\mathcal{P}_{\Lambda}^n}a_n^T(S_1,\ldots,S_n)\prod_{i=1}^nw_{\beta,h}(S_i),
	\end{align}
	with the Ursell function $a_n^T(\cdot)$ defined as 
	\begin{equation}\label{Ursell.F.X1}
		a^T_n(S_1,\ldots, S_n):= \sum_{G\in\mathcal{C}[n]}\prod_{\{i,j\}\in E(G)}[\zeta(S_i, S_j)-1].
	\end{equation}
	Expression \eqref{eq:Is.X.mag.stronggg16} is the well-known cluster expansion.
	
	The next theorem establishes a sufficient condition for the existence of the pressure function as $\Lambda \to \mathbb{Z}^d$ in the thermodynamic limit and allows us to verify the analytic domain of the pressure function.
	\begin{Thm}\label{Thm.Xuan.strong.1} If there exists $a>0$ such that
		\begin{align}\label{eq.Xuan.new.GK.1111}
			\sup_{x\in\mathbb Z^d}\sum_{x\in S\atop S\in\mathcal{P}}w_{\beta,h}(S)\eee^{a|[S]_1|}\;\le\; \eee^a-1
		\end{align}
		with
		\begin{equation}\label{eq.Xuan.new.1111}
			[S]_1\;:=\;\{j\in \mathbb{Z}^d:d(j,S)\le 1\},
		\end{equation}
		then the following hold:
		
		(i.) Denote $\bbd{w}_{\beta,h}:=\{w_{\beta, h}(S)\}_{S\in \mathcal{P}_{\mathbb Z^d}}$. For a finite subset $S\in\mathcal{P}_{\mathbb Z^d}$, $\card{\Gamma}_S(\bbd{w}_{\beta,h})$ defined in
		\begin{align}\label{eq.Xuan.Robcor13}|\Gamma|_S(\bbd{w}_{\beta,h})=1+\sum_{n=1}^{\infty}\frac{1}{n!}\sum_{(S_1,\ldots, S_n)\in\mathcal{P}_{\mathbb V}^n}|a^T_{n+1}(S,\ldots,S_n)|\prod_{i=1}^{n}w_{\beta,h}({S_i})
		\end{align}
		converges. Furthermore, for $S\in\mathcal{P}_{\mathbb Z^d}$,
		\[\card{\Gamma}_S(\bbd{w}_{\beta,h})\le \eee^{a|S|}.\]
		
		(ii.) The free energy function  \eqref{eq:Is.mag.strong.16} converges absolutely and uniformly in $\Lambda$, and 
		\begin{equation}\label{eq.Xuan.new.12} p^+(\beta, h)=\beta d+h+\sum_{X\subset \mathbb Z^d: X\ni 0}\frac{1}{|X|}\Psi(X),\end{equation}
		where for each $X\subset \mathbb Z^d$, $\Psi(\cdot)$ is defined as follows
		\begin{align}
			\Psi(X)=\sum_{n=1}^{\infty}\frac{1}{n!}\sum_{(S_1\ldots S_n)\in\mathcal{P}_{\mathbb Z^d}^n\atop S_1\cup \cdots \cup S_n=X}a_n^T(S_1,\ldots,S_n)\prod_{i=1}^nw_{\beta,h}(S_i)\nonumber
		\end{align}
	\end{Thm}
	As stated in the introduction, our primary goal in this subsection is to determine the domain of the inverse temperature \(\beta\) and the external magnetic field \(h\) for which the pressure function \(p(\beta, h)\) is analytic. This is outlined in the following theorem.
	\begin{Thm}\label{main.result.X1} The pressure function $p(\beta,h)$ is analytic in the domain $\mathcal{D}$ with
		\[\mathcal{D}\;=\;\left\{(\beta, h)\in\mathbb R\times \mathbb R:2h
		\;\ge\;\varphi^{\rm st}(d)\right\}\]
		where $\varphi^{\rm st}(d)$ is defined by
		\begin{equation}\label{eq.min.strong.RX.1}
			\varphi^{\rm st}(d):=\min_{a> 0}\phi^{\rm st}(a)=(4d+1)\log\left[1+\frac{1}{4d}\right]+\log(4d),
		\end{equation}
		with 
		\[\phi^{\rm st}(a):=(4d+1)a-\log(\eee^a-1).\]
	\end{Thm}
	The proof of Theorem \ref{Thm.Xuan.strong.1} and Theorem \ref{main.result.X1} follow from cluster expansion theory, which is discussed in more detail in Subsection \ref{Sub.Xuan.new.1}. 
	
	\medskip 
	\paragraph{\bf Comparison of analyticity domains.} We compare our estimations with the results obtained by Friedli and Velenik in their book \cite{FV17}, which presents the standard established findings on this topic. To describe their result, we define
	\begin{eqnarray}
		\eta(h,d) &=& \sum_{k=1}^\infty (2d)^{2k} \eee^{(2d+1 - 2h)k}, \nonumber \\
		H^+ &=& \{ h : \mathrm{Re}\, h \ge \bar{h} \}
	\end{eqnarray}
	where
	\[
	\bar{h} := \inf \{ h > 0 : \eta(h, d) < 1 \}.
	\]
	According to Friedli and Velenik, if $h \in H^+$, then the pressure function $p^+(\beta, h)$ is analytic. A straightforward computation yields the estimate
	\begin{equation}
		2\bar{h} = \log(8d^2) + 2d + 1.
	\end{equation}
	By substituting $a = \log(1+1/(4d+1))$ into the function $\phi^{\rm st}(a)$, we obtain
	\begin{align} \label{eq.Xuan.com.1}
		2\bar{h} = \log(8d^2) + 2d + 1 
		\ge (4d+1)\log\left(1+\frac{1}{4d+1}\right)+\log(4d+1)
		\ge \min_{a > 0} \phi^{\rm st}(a)=\varphi^{\rm st}(d).
	\end{align}
	Inequality~\eqref{eq.Xuan.com.1} indicates that our estimate is less restrictive than the bound established by Friedli and Velenik~\cite{FV17} (see Figure~\ref{Fig.1} for more details). To compare the two bounds, we define the ratio
	\[
	r^{\rm st}(d) := \frac{\varphi^{\rm st}(d)}{\log(8d^2) + 2d + 1},
	\]
	where $\varphi^{\rm st}(d)$ is defined in \eqref{eq.min.strong.RX.1}. Figure~\ref{Fig.2} shows that this ratio decays exponentially, tending to zero as $d \to \infty$.
	\begin{figure}[h]
		\centering
		\includegraphics[width=0.5\linewidth]{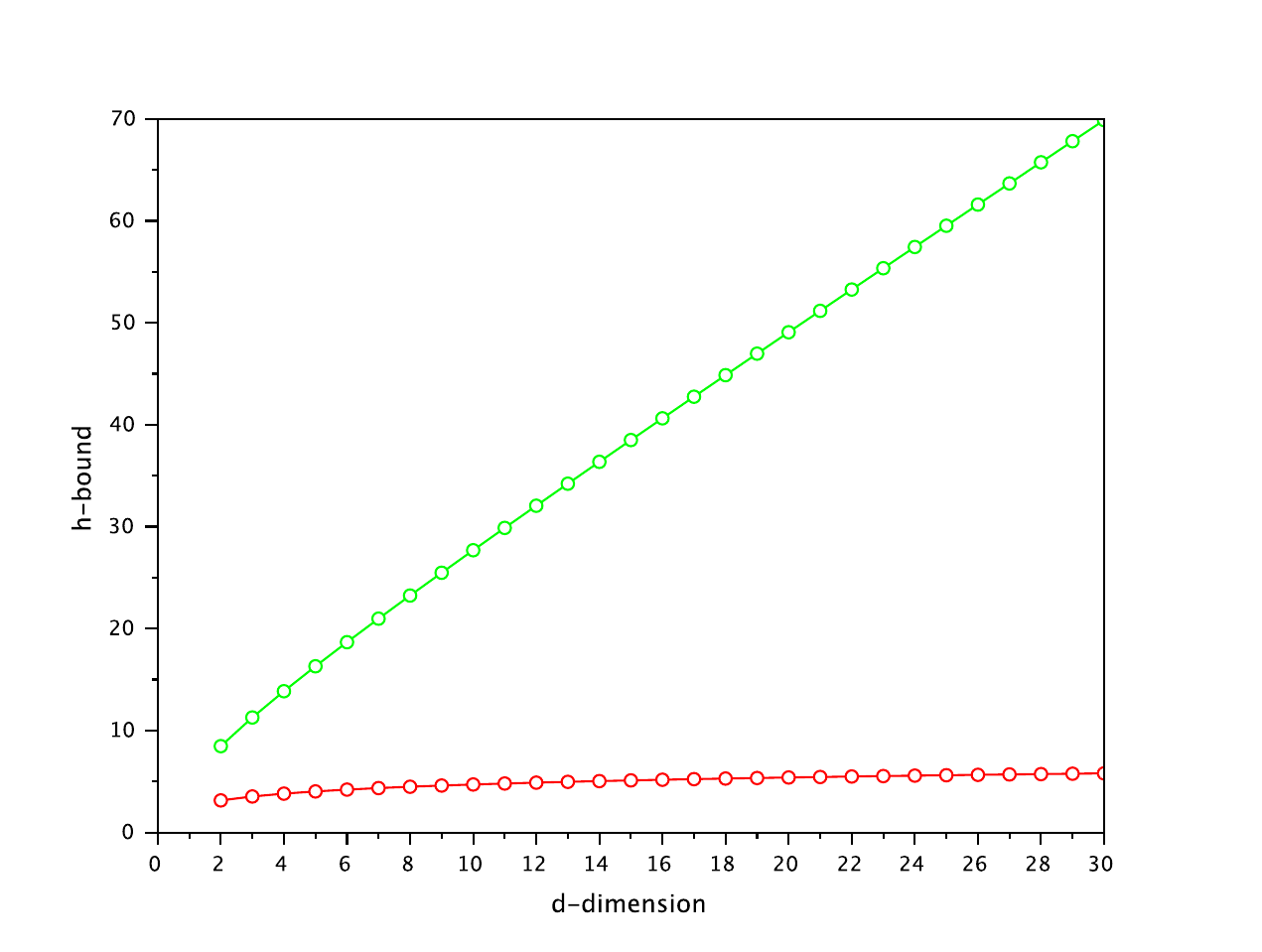}
		\caption{A comparison with the result of  Friedli and Velenik (The green and red lines  represent the Friedli--Velenik bound and  our result, respectively.)}
		\label{Fig.1}
	\end{figure}
	\begin{figure}[h]
		\centering
		\includegraphics[width=0.5\linewidth]{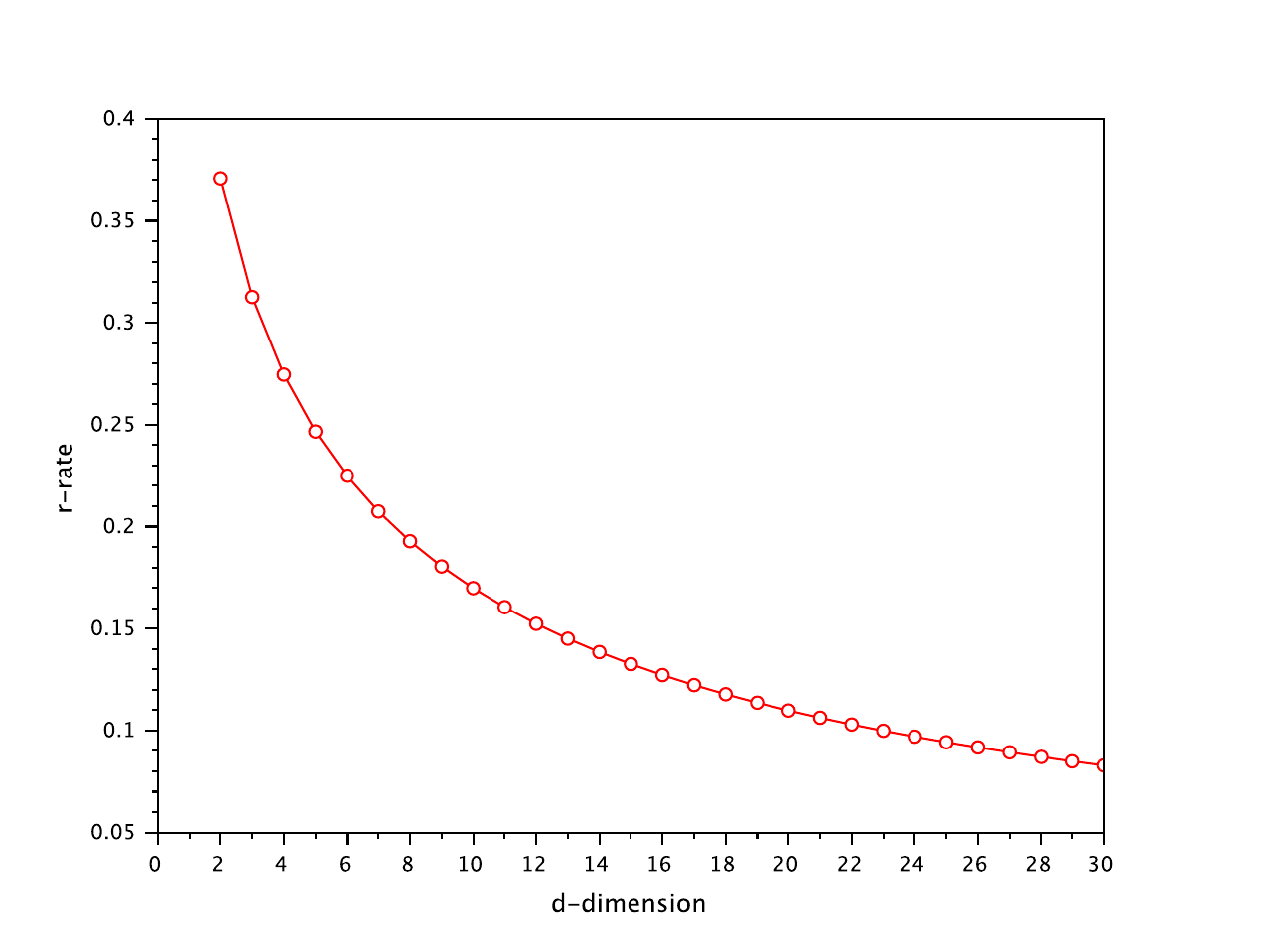}
		\caption{The rate between our bound and the Friedli--Velenik bound}
		\label{Fig.2}
	\end{figure}

	\subsection{The Ising model at high-temperature without magnetic field $(h=0)$}\label{subsec.X.high.tem1}
	In the case of high temperatures and a vanishing external field, for simplicity in computation, we consider the Ising model with free boundary conditions. Its partition function is defined as:
	\begin{equation}\label{eq:Is.mag.high1}
		Z^{\varnothing}_{\Lambda}(\beta,0)\;=\;\sum_{\sigma_{\Lambda}}\exp\Bigg\{\beta\sum_{x,y\in\Lambda\atop \card{x-y}=1}\sigma_x\sigma_y\Bigg\}.
	\end{equation}
	We use the identity
	\begin{equation}\label{eq:Is.mag.high2}
		\exp\{\beta\sigma_x\sigma_y\}\;=\;\cosh\beta+\sigma_x\sigma_y\sinh\beta,
	\end{equation}
	to rewrite it in  the form
	\begin{eqnarray}\label{eq:Is.mag.high3}
		Z^{\varnothing}_{\Lambda}(\beta,0)&=&(\cosh\beta)^{\card{\mathcal{E}_{\Lambda}}}\sum_{E\subset \mathcal{E}_{\Lambda}}
		(\tanh\beta)^{\card{E}}\sum_{\sigma\in\Omega_{\Lambda}}\prod_{\{i,j\}\in E}\sigma_i\sigma_j\nonumber\\&=&(\cosh\beta)^{\card{\mathcal E_{\Lambda}}}\sum_{E\subset \mathcal E_{\Lambda}}
		(\tanh\beta)^{|E|}\sum_{\sigma\in\Omega_{\Lambda}}\prod_{i\in\Lambda}\sigma_i^{I(i,E)}\nonumber\\&=&(\cosh\beta)^{\card{\mathcal{E}_{\Lambda}}}\sum_{E\subset \mathcal E_{\Lambda}}
		(\tanh\beta)^{\card{E}}\prod_{i\in\Lambda}\bigg[\sum_{\sigma_i=\pm 1}\sigma^{I(i,E)}\bigg]\nonumber\\&=&2^{|\Lambda|}(\cosh\beta)^{|\mathcal E_{\Lambda}|}\sum_{E\subset \mathcal{E}_{\Lambda}}
		(\tanh\beta)^{|E|}\prod_{i\in\Lambda}\one{i\in\Lambda:\;I(i,E)\;  \rm{ even}}
	\end{eqnarray}
	where $I(i,E)\;=\;\card{\big\{j\in\mathbb{Z}^d:\{i,j\}\in E\big\}}$ is the incidence number and $\mathcal{E}_{\Lambda}$ is defined as in \eqref{eq.Xuan.edge.1}.
	Let us denote 
	\begin{eqnarray}\label{eq:Is.mag.high4}
		\mathcal{E}_{\Lambda}^{\rm{even}}&=&\big\{E\subset \mathcal{E}^{b}_{\Lambda}:I(i,E)\text{ is even for all }i\in\Lambda\big\}.
	\end{eqnarray}
	To improve our understanding of the convergence conditions in comparison to previous findings, such as those presented by Simon in \cite{Simon93} or the recent results shared by Aldo Procacci in his unpublished lecture notes \cite{AP23}, we would like to revisit a lesser-known result in graph theory known as Veblen's theorem concerning about a property  of a connected graph $E\in\mathcal{E}^{\rm even}_{\Lambda}$. This theorem was established by Veblen in \cite{Veb12}. 
	\begin{lem}[Veblen's Theorem]\label{thm.Veb} The set of edges of a finite connected graph can be written as a union of meaning edge-disjoint simple cycles if and only if every vertex has an even number of incident edges.
	\end{lem}
	Each set \( E \subset \mathcal{E}_{\Lambda}^{\text{even}} \) can be represented as a graph, where the vertices correspond to the endpoints of the edges. This graph can then be decomposed into maximally connected components. This decomposition results in a partition of the set of edges, expressed as \( E = E_1 \cup \cdots \cup E_n \), where each \( E_i \in \mathcal{E}^{\text{even}}_{\Lambda} \). According to Lemma \ref{thm.Veb}, each set \( E_i \) consists of edges that form a cycle in the lattice \( \mathbb{Z}^d \) for \( i = 1, \ldots, n \). Before proceeding with further calculations, let us further decompose  each cycle into a primitive cycle based on the following definition.
	\begin{defn}
		A cycle $\mathfrak{c}$ is called \emph{primitive} if it cannot be partitioned into two cycles $\mathfrak{c}^{\prime}$ and $\mathfrak{c}^{\dprime}$. 
	\end{defn}
	If $\mathfrak{c}$ is not primitive then it can be partitioned into two parts $\mathfrak{c}^{\prime}$ and $\mathfrak{c}^{\dprime}$ with no common edge but sharing some common sites. In particular, $|\mathfrak{c}|=|\mathfrak{c}^{\prime}|+|\mathfrak{c}^{\dprime}|$. A primitive cycle is also referred to as a polygon in $\mathbb{Z}^d$. We shall find it convenient to employ $\mathcal{P}^
	{\rm p}_{\mathbb Z^d}, \mathcal{P}^{\rm p}_{\Lambda}$ as the collection of  closed polygons in $\mathbb{Z}^d$ and $\Lambda$ respectively. Therefore, each set $E\in\mathcal{E}^{\rm even}$ can be decomposed into 
	\[E=\bigcup_{i=1}^n \mathfrak{p}_i\]
	where each $\mathfrak{p}_i\in\mathcal{P}^{\rm p}_{\Lambda}$.
	We then can rewrite 
	the right-hand side of the last equation in \eqref{eq:Is.mag.high3} in the form
	\begin{equation} \label{eq:Is.mag.high4} 
		Z_{\Lambda}(\beta,0)=2^{\card{\Lambda}}(\cosh\beta)^{|\mathcal E_{\Lambda}|}\,\Xi^{\rm{HT}}_{\Lambda}(\beta)
	\end{equation}
	with
	\begin{equation}\label{eq:Is.mag.high5}
		\Xi^{\rm{HT}}_{\Lambda}(\beta)=1+\sum_{n\ge 1}\frac{1}{n!}\sum_{(\mathfrak{p}_1\ldots,\mathfrak{p}_n)\in \left[\mathcal{P}^{\rm{p}}_{\Lambda}\right]^n}\prod_{1\le i< j\le n}\zeta(\mathfrak{p}_i, \mathfrak{p}_j)\prod_{i=1}^n(\tanh\beta)^{|\mathfrak{p}_i|}.
	\end{equation}
	with $\zeta(\cdot)$ function defined as in \eqref{eq.Ising.model.Xuan.1} and 
	\begin{equation}
		w^{\rm high}(\mathfrak{p})\;=\;\left[\tanh\beta\right]^{\card{\mathfrak{p}}}.
	\end{equation}
	By using the cluster expansion theory, the pressure with $\varnothing$-boundary condition in $\Lambda$ can be expressed as the following form:
	\begin{equation}\label{eq:Is.high.16}
		P^{\varnothing}_{\Lambda}(\beta)=\beta\frac{|\mathcal{E}_{\Lambda}|}{|\Lambda|}+\frac{1}{|\Lambda|}\log \Xi^{\rm HT}_{\Lambda}(\beta,h),
	\end{equation}
	where
	\begin{equation}\label{eq:Is.highhh16}
		\log \Xi^{\rm{HT}}_{\Lambda}(\beta)=\sum_{n=1}^{\infty}\sum_{(\mathfrak{p}_1,\ldots,\mathfrak{p}_n)\in\left[\mathcal{P}^{\rm{p}}_{\Lambda}\right]^n}a_n^{T}(\mathfrak{p}_1,\ldots,\mathfrak{p}_n)\prod_{i=1}^nw_\beta^{\rm high}(\gamma_i)
	\end{equation}
	with the Ursell function $a_n^T(\cdot)$ defined in  \eqref{Ursell.F.X1}.
	
	In the following two theorems \ref{Thm.Xuan.hightem} and \ref{main.result.high.X1}, we present a condition for the existence of the pressure function as \(\Lambda\to\mathbb{Z}^d\) and give a full representation of the pressure function through the thermodynamic limit. Based on the stated condition, we will establish the analytic domain for the pressure function in the final theorem of this subsection.
	\begin{Thm}\label{Thm.Xuan.hightem} If there exists $a>0$ such that
		\begin{align}\label{eq.Xuan.new.GK.high.1111}
			\sup_{{\rm q}\in\mathcal{E}_{\mathbb{Z}^{d}}}\sum_{{\rm q}\in \mathfrak{p}\atop \mathfrak{p}\in\mathcal{P}^{\rm p}_{\mathbb{Z}^d}}w^{\rm high}_{\beta}(\mathfrak{p})(\coth a)^{|\mathfrak{p}|}\;\le\; \coth a-1,
		\end{align}
		then the following hold:
		
		(i) Denote $\bbd{w}_\beta^{\rm high}:=\{w_\beta^{\rm high}(\mathfrak{p})\}_{\mathfrak{p}\in \mathcal{P}^{\rm p}_{\mathbb Z^d}}$. For each finite polygon $\mathfrak{p}\in\mathcal{P}_{\mathbb Z^d}^{\rm p}$, $\card{\Gamma}_{\mathfrak{p}}(\bbd{w}^{\rm high}_{\beta})$, defined in
		\begin{align}\label{eq.Xuan.Robcor113}|\Gamma|_{\mathfrak{p}}(\bbd{w}^{\rm high}_{\beta})=1+\sum_{n=1}^{\infty}\frac{1}{n!}\sum_{(S_1,\ldots, S_n)\in\mathcal{P}_{\mathbb V}^n}|a^T_{n+1}(S,\ldots,S_n)|\prod_{i=1}^{n}w_{\beta}^{\rm high}({\mathfrak{p}_i})
		\end{align} converges. Furthermore, each finite polygon $\mathfrak{p}\in\mathcal{P}_{\mathbb Z^d}^{\rm p}$,
		\[\card{\Gamma}_{\mathfrak{p}}(\bbd{w}^{\rm high}_{\beta})\le \eee^{a|\mathfrak{p}|}.\]
		
		(ii) The free energy function  \eqref{eq:Is.high.16} converges absolutely and uniformly in $\Lambda$, and for a fixed edge $e\in\mathcal{E}_{\mathbb Z^d}$. 
		\begin{equation}\label{eq.Xuan.new.high.12} p^\varnothing(\beta)=\log 2+d\cosh\beta +d\sum_{E\subset \mathcal{E}_{\mathbb Z^d}: E \ni e}\frac{1}{|E|}\Psi^{\rm high}(E),\end{equation}
		where, for each finite set of edges $E\subset  \mathcal{E}_{\mathbb Z^d}$, $\Psi^{\rm high}(\cdot)$ is defined as follow:
		\begin{align}
			\Psi^{\rm high}(E)=\sum_{n=1}^{\infty}\frac{1}{n!}\sum_{(\mathfrak{p}_1,\ldots ,\mathfrak{p}_n)\in [\mathcal{P}^{\rm p}_{\mathbb Z^d}]^n\atop \mathfrak{p}_1\cup \cdots \cup \mathfrak{p}_n=E}a_n^T(\mathfrak{p}_1,\ldots,\mathfrak{p}_n)\prod_{i=1}^nw^{\rm high}_{\beta}(\mathfrak{p}_i).
		\end{align}
	\end{Thm}
	One of our main contributions in this paper is presented in the following theorem.
	\begin{Thm}\label{main.result.high.X1} The pressure function $p^\varnothing(\beta)$ is analytic in the domain $\mathcal{D}$ with
		\[\mathcal{D}\;=\;\{\beta\in\mathbb R:\beta\le\;\tanh^{-1}[\phi_1^{\rm high}(\bar{a}^{\rm high})]\}\]
		where $\phi^{\rm high}_1$ is defined by 
		\begin{equation}\label{eq:rr.mag.RXhigh11}\phi_1^{\rm high}(a):=\frac{\tanh a}{\sqrt{2d-1}}\sqrt{\frac{(1-\coth a)+\sqrt{(\coth a-1)^2+\frac{8}{2d-1}(\coth a-1)}}{4}},
		\end{equation}
		and $\bar{a}^{\rm high}$ represents the value of $a$ at which $\phi^{\rm high}_1$ reaches its maximum.
	\end{Thm}
	As discussed in Subsection \ref{subsection.Ising.strong.Xuan1},  Theorem \ref{Thm.Xuan.hightem} and  Theorem \ref{main.result.high.X1} are explained in more detail in Subsection \ref{sec.app.clu2}. 
	
	\medskip
	\paragraph{\bf Comparison of analyticity domains.} Let us compare  our result in Theorem \ref{main.result.high.X1} with the best published result provided by B. Simon in \cite[Chapter V]{Simon93}.  In this reference, the analyticity holds if 
	\begin{equation}\label{eq:Is.mag.high14}
		\beta\exp \beta\;\le\;\frac{1}{48d^2}\;=:\; A_S,
	\end{equation}
	which yields the bound
	\begin{equation}\label{eq:Is.mag.high14.1}
		\beta\;\le\; \beta_S\;:=\; W(A_S),
	\end{equation}
	with $W$ being the Lambert function.
	From condition \eqref{eq:rr.mag.RXhigh11}, we can derive a weaker condition for the analyticity of the pressure function $p^{\varnothing}(\beta,0)$ as follows
	\begin{align}
		\tanh\beta&\le \max_{a>0}\frac{\tanh a}{2d-1}\sqrt{\frac{8(\coth a-1)}{\coth a-1+\sqrt{(\coth a-1)^2+8/3(\coth a-1)}}}\nonumber\\ &=\frac{1.0783}{2d-1}\nonumber\\ &\le \phi_1^{\rm high}(\bar{a}^{\rm high}).
	\end{align}
	Therefore, we obtain
	\begin{equation}\label{eq:Is.mag.high15}
		\beta\;\le \;\frac{1}{2}\log\frac{1+\frac{1.0783}{2d-1}}{1-\frac{1.0783}{2d-1}}\;\le\; \beta_N.
	\end{equation}
	where \[\beta_N:=\tanh^{-1}[\phi_1^{\rm high}(\bar{a}^{\rm high})]=\frac{1}{2}\log \frac{1+\phi^{\rm high}(\bar{a}^{\rm high})}{1-\phi^{\rm high}(\bar{a}^{\rm high})}.\]
	To compare with Simon's result, we start with the inequalities
	\begin{align}\beta_N\ge \frac{1}{2}\log \frac{1+\frac{1.0783}{2d-1}}{1-\frac{1.0783}{2d-1}}\ge \frac{1}{2}\left[\frac{1.0783}{2d-1}-\frac{(1.0783)^2}{2(2d-1)^2}+\frac{1.0783}{2d-1}\right]=\frac{1.0783}{2d-1}-\frac{(1.0783)^2}{4(2d-1)^2}:=\beta_c.
	\end{align}
	Then, for $d\ge 2$, we have
	\begin{align}
		\beta_N\eee^{\beta_N}\ge \beta_c\eee^{\beta_c}\ge \beta_c&> \frac{1}{2d-1}-\frac{3}{8(2d-1)^2}+\frac{1}{32d^2}-\frac{1}{8(2d-1)^2}\nonumber\\ &> \frac{1}{2d-1}-\frac{1}{2(2d-1)^2}+\frac{1}{32d^2}=\frac{1}{2d-1}\left[1-\frac{1}{4d-2}\right]+\frac{1}{32d^2}> \frac{1}{32d^2}>\frac{1}{48d^2}.
	\end{align}
	It implies that our domain includes the domain proposed by Simon. To clarify further, let us examine the rate function \( r^{\text{high}} \) that measures the ratio between our bound and Simon's bound, defined as follows
	\begin{equation}
		r^{\rm high}(d):=\frac{\beta_N}{\beta_S}=\frac{\displaystyle\frac{1}{2}\log(1+\phi_1^{\rm high}(\bar{a}^{\rm high}))-\frac{1}{2}\log(1-\phi_1^{\rm high}(\bar{a}^{\rm high}))}{\displaystyle W\left(\frac{1}{48d^2}\right)}.
	\end{equation}
	We observe that our bound is significantly stronger than Simon's bound, as illustrated in Figure \ref{Fig.3} and \ref{Fig.5}.
	\begin{figure}[h]
		\centering
		\includegraphics[width=0.5\linewidth]{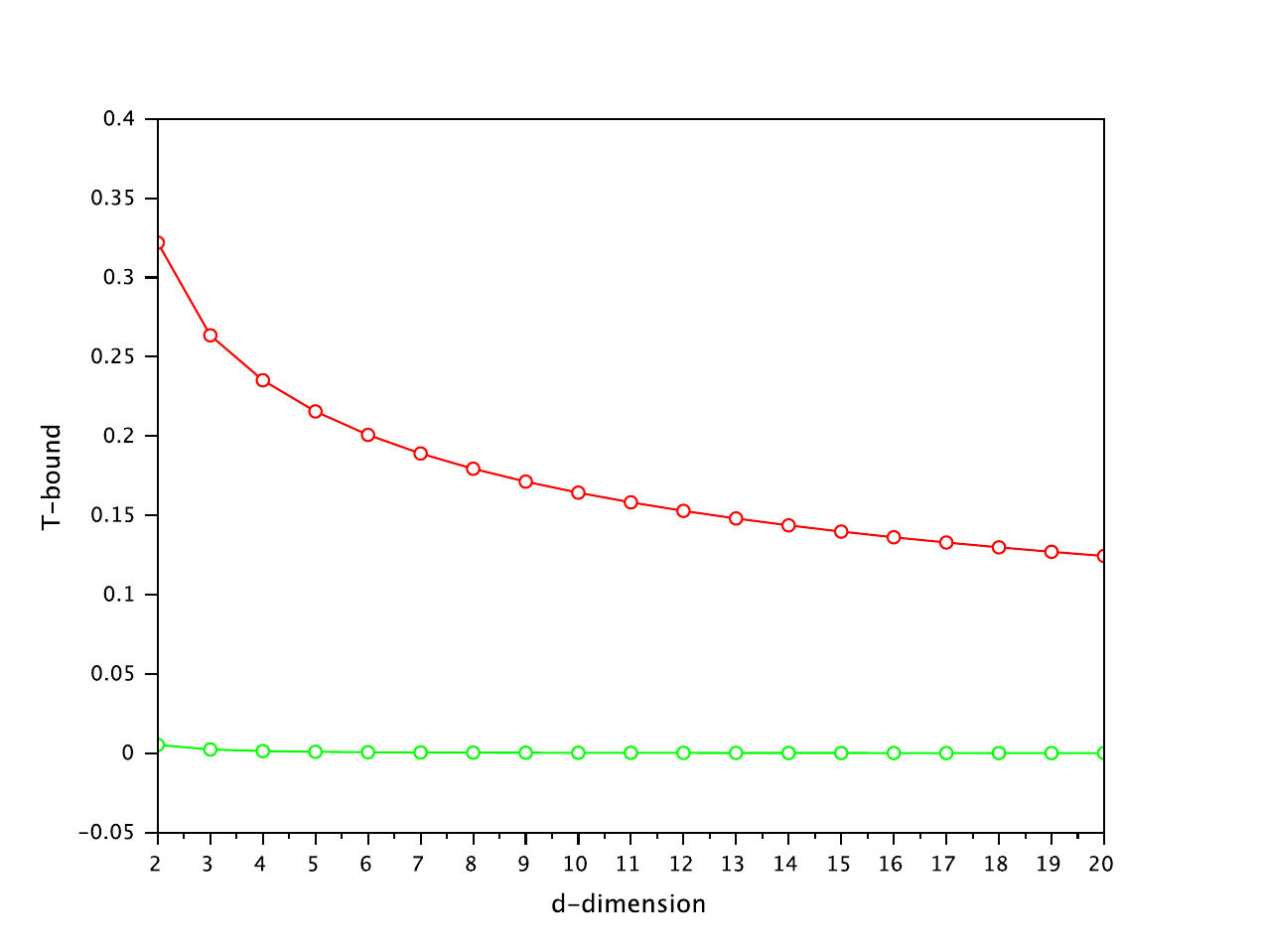}
		\caption{A comparison with Simon's result (green line presented for Simon's result and red line presented for our result)}
		\label{Fig.3}
	\end{figure}
	\begin{figure}[h]
		\centering
		\includegraphics[width=0.5\linewidth]{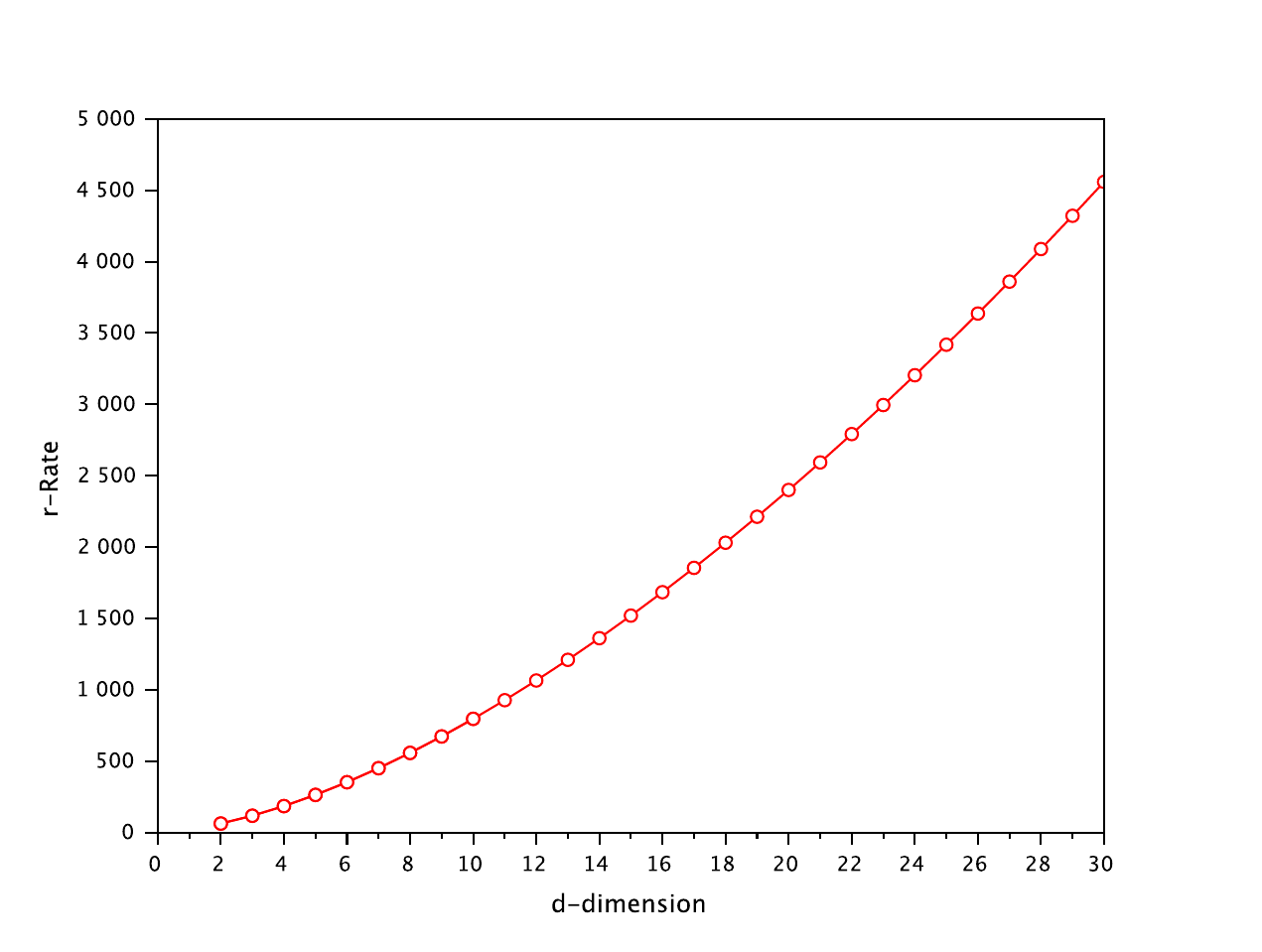}
		\caption{The rate between our bound and Simon's bound}
		\label{Fig.5}
	\end{figure}
	
	\noindent
	As noted in \cite{AP23}, Procacci established an analyticity domain for \( p^{\varnothing}(\beta,0) \) in the case \( d=2 \), which is valid for \( \beta \le 0.151 \). However, Procacci's bound is considerably weaker than our result, which extends the limit to \( \beta \le \beta_N(2)=0.322 \).
	
	\subsection{The Ising model at low-temperature without magnetic field ($h=0$)}\label{subsec:Ising.model.low}
	In this subsection, we build upon the work of J. L. Lebowitz and A. E. Mazel, as referenced in \cite{LM98}.
	Let us consider zero magnetic field and ``plus'' boundary condition: $\omega_i=1$, for all $i\notin \Lambda$. By employing a similar method as  in Subsection \ref{subsection.Ising.strong.Xuan1}, we add and subtract $1$ from each term in this Hamiltonian. Then the  Hamiltonian can be redefined as
	\begin{equation}
		\mathrm{H}^+_{\Lambda;\beta,0}(\bbd{\sigma}_\Lambda)=-\beta |\mathcal{E}_{\Lambda}|+\sum_{\bbd{\sigma}_{\Lambda}\in\Omega_\Lambda}|\partial\Lambda^-(\bbd{\sigma}_{\Lambda})|
	\end{equation}
	We now wish to give a geometrical description of $\partial\Lambda^-(\bbd{\sigma}_{\Lambda})$ for  each configuration $\bbd{\sigma}_{\Lambda}\in \Omega^+_{\Lambda}$ that can better account for the low-temperature trend towards alignment of nearest-neighbor spins. The starting point is thus to express the Hamiltonian in a form that emphasizes the role played by pairs of opposite spins.
	
	We associate to each $i\in\mathbb{Z}^d$ the closed unit cube of $\mathbb{R}^d$ centered at $i$:
	\[\mathbf{J}_i\;=\;i+\bigg[\frac{-1}{2},\frac{1}{2}\bigg]^d.\]
	The boundary of $\mathbf{J}_i$, in the sense of the standard topology on $\mathbb{R}^d$, denoted by $\partial\mathbf{J}_i$. The \emph{dual lattice} is defined as
	\[\mathbb{Z}_*^d\;:=\;\mathbb{Z}^d+\underbrace{\left(\frac{1}{2},\ldots,\frac{1}{2}\right)}_{d-\text{dimensional}}.\]
	Let us start with some basic definitions  in elementary geometry. 
	\begin{defn}
		A \emph{plaquette} is a unit $(d-1)$-dimensional face of a $d$-dimensional hypercube $\mathbf{J}_i$, $i\in\mathbb{Z}^d$. Let
		$\mathbb E_{\mathbb Z^d}$ be the set of  all plaquettes ${\rm p}$  in $\mathbb Z^d$.
	\end{defn}
	\begin{defn}[Adjacency relations]
		We use the following notions of adjacency.
		\begin{enumerate}
			\item  Two plaquettes are  \emph{adjacent} if they have common $(d-2)$-dimensional face.
			\item Two lattice sites are called \emph{adjacent} if they are the endpoints of a lattice bond.
			\item A plaquette and a lattice site are \emph{adjacent} if this plaquette intersects one of the lattice bonds incident on this site.
			\item  A lattice site and a $(d-2)$-dimensional face are \emph{adjacent} if this site is adjacent to one of four plaquettes incident on this face.
		\end{enumerate}
	\end{defn}
	\begin{defn}
		A set of plaquettes is \emph{connected} if any two of its plaquettes belong to a chain of pairwise adjacent plaquettes from the set. 
	\end{defn}
	\begin{defn}[Contour]
		A \emph{contour} (or \emph{Peierls contour}) is a connected closed set of plaquettes that separates
		$\mathbb Z^d$ into exactly two connected components: a finite component
		$\Omega_\gamma$, called the \emph{interior} of $\gamma$, and an infinite component $\Omega_\gamma^c=\mathbb Z^d\setminus\Omega_\gamma$, called the \emph{exterior} of $\gamma$.
		
		The size of a contour $\gamma$, denoted by $|\gamma|$, is the number of plaquettes contained in $\gamma$.
	\end{defn}
	
	\begin{defn}
		A contour $\gamma$ is called  \emph{primitive} if it can not be partitioned into two contours $\gamma^{\prime}$ and $\gamma^{\dprime}$. We shall find it convenient to employ $\Gamma^{\mathbb Z^d}_{\rm{prim}}, \Gamma^{\Lambda}_{\rm{prim}}$ as  the collection of primitive contours in $\mathbb{Z}^d$ and $\Lambda$ respectively.
	\end{defn}
	If $\gamma$ is not primitive then it can be partitioned into two parts $\gamma^{\prime}$ and $\gamma^{\dprime}$ with no common plaquettes but sharing some $(d-2)$--dimensional faces. In particular, $|\gamma|=|\gamma^{\prime}|+|\gamma^{\dprime}|$.
	
	To sum over contours we follow the approach of J. L. Lebowitz and A. E. Mazel in \cite{LM98}.  Fix a plaquette ${\rm p}$. Every contour $\gamma\in\Lambda$ containing ${\rm p}$ can be uniquely decomposed into maximal primitive subcontours. The decomposition can be naturally endowed with a rooted tree-like structure. The root of the tree is the primitive contour $\gamma_{0}$ which contains the plaquettes ${\rm p}$. The first generation is formed by all primitive subcontours $\gamma_{1,i_1}$ of $\gamma$ which have common $(d-2)$--dimensional faces with $\gamma_0$. Subsequently, the 
	subcontours of the $n^{\rm th}$ generation are the primitive subcontours  $\gamma_{n,i_n}$  which have a common $(d-2)$--dimensional face with some of primitive subcontour from generation $n-1$ and are not included in $\bigcup_{k=1}^{n-1}\bigcup_{i_k}\gamma_{k,i_k}$. 
	
	For a configuration $\bbd{\sigma}_{\Lambda}$ containing a finite number of sites $i$, we denote
	\begin{equation}
		\mathcal{M}(\bbd{\sigma}_{\Lambda})\;=\;\bigcup_{i\in\Lambda^-(\bbd{\sigma}_{\Lambda})}\mathbf{J}_i
	\end{equation}
	then $\partial \mathcal M(\bbd{\sigma}_{\Lambda})$ is  made of plaquettes of the dual lattice that are orthogonal to ``frustrated'' bonds.  That is, if we denote $\{i,j\}_{\bot}$ the plaquette orthogonal to $\{i,j\}$, then each $\{i,j\}_\bot\in\partial \mathcal{M}(\bbd{\sigma}_{\Lambda})$ separates two opposite spins $\sigma_i\ne\sigma_j$. We obtain
	\begin{equation}
		\mathrm{H}^+_{\Lambda;\beta,0}(\bbd{\sigma}_\Lambda)=-\beta |\mathcal{E}_{\Lambda}|+\sum_{\bbd{\sigma}_{\Lambda}\in\Omega_\Lambda}|\partial\mathcal{M}(\bbd{\sigma}_{\Lambda})|.
	\end{equation}
	It is not hard to see that, due to the plus boundary conditions, $\partial\mathcal{M}(\bbd{\sigma}_{\Lambda})$ is a closed surface consisting of several primitive contours (see Figure \ref{fig:Ising_models2}), i.e.
	\[\partial \mathcal{M}(\omega)=\gamma_1\cup\cdots\cup \gamma_n.\]
	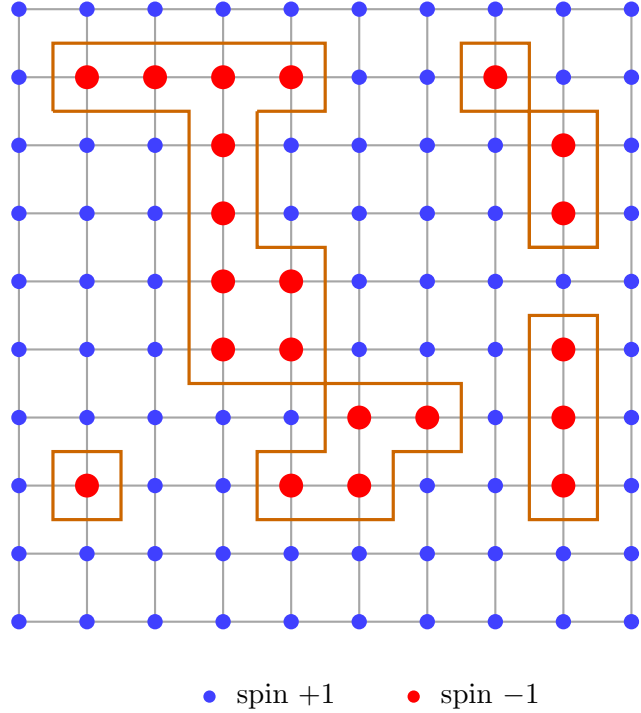
\begin{figure}[hbt!]
		\begin{tikzpicture}[scale=0.9]
			
			\def\N{9}
			
			\draw[gray!70,line width=0.8pt]
			(0,0) grid (\N,\N);
			
			\foreach \i in {0,...,\N}
			{
				\foreach \j in {0,...,\N}
				{
					\fill[blue!75]
					(\i,\j) circle (3.2pt);
				}
			}
			
			
			\foreach \x/\y in {
				2/8, 3/8, 4/8, 3/7, 3/6, 4/5, 3/4, 4/4, 5/3, 3/5, 6/3, 1/2, 1/8, 8/2, 4/2, 5/2, 8/6, 8/4,8/3, 7/8, 8/7
			}
			{
				\fill[red]
				(\x,\y) circle (5pt);
			}
			
			
			\draw[orange!80!black,line width=1.2pt]
			(0.5,7.5)--(0.5,8.5)--(4.5,8.5)--(4.5,7.5)--(3.5,7.5);
			
			\draw[orange!80!black,line width=1.2pt]
			(3.5,7.5)--(3.5,5.5)--(4.5,5.5)--(4.5,3.5)
			--(6.5,3.5)--(6.5,2.5)--(5.5,2.5)--(5.5,1.5)
			--(3.5,1.5)--(3.5,2.5)--(4.5,2.5)--(4.5,3.5)--(2.5,3.5)--(2.5,7.5)--(0.5,7.5);
			
			\draw[orange!80!black,line width=1.2pt]
			(7.5,7.5)--(6.5,7.5)--(6.5,8.5)--(7.5,8.5)--(7.5,5.5)--(8.5,5.5)--(8.5,7.5);
			
			\draw[orange!80!black,line width=1.2pt]
			(7.5,5.5) rectangle (8.5,7.5);
			
			\draw[orange!80!black,line width=1.2pt]
			(7.5,1.5) rectangle (8.5,4.5);
			
			\draw[orange!80!black,line width=1.2pt]
			(0.5,1.5) rectangle (1.5,2.5);


			\fill[blue!75] (2.8,-1.1) circle (0.09);
			\node[right] at (3.05,-1.1) {spin $+1$};
			
			\fill[red] (5.8,-1.1) circle (0.09);
			\node[right] at (6.05,-1.1) {spin $-1$};
			
		\end{tikzpicture}
		\caption{A configuration of the two-dimensional Ising model in a finite box $\Lambda$ with $+$ boundary condition. At low temperature, the lines separating regions of $+$ and $-$ spins are expected to be short and sparse, leading to a positive magnetization in $\Lambda$.}
		\label{fig:Ising_models2}
	\end{figure}

	Let \( V_\gamma \) be the set containing all sites \( i \in \mathbb{Z}^d \) that are in the interior of \( \gamma \) and adjacent to the contour \( \gamma \).
	\begin{defn}
		Two primitive contours $\gamma$ and $\gamma^{\prime}$ are \emph{compatible} if they do not share any common plaquettes and $V_{\gamma}\cap V_{\gamma'}=\varnothing$ (they may, however, share common $(d-2)$-dimensional faces). 
		This situation is denoted by $\gamma\sim\gamma^{\prime}$; otherwise the contours are
		\emph{incompatible} and the relation is denoted by $\gamma\nsim\gamma^{\prime}$
	\end{defn}
	We then can write 
	\begin{equation}\label{eq:Is.mag.low3} 
		Z^+_{\Lambda}(\beta,0)\;=\;\eee^{\beta \card{\mathcal E_{\Lambda}
		}}\,\Xi^{\rm{LT}}_{\Lambda}(\beta),
	\end{equation}
	where the \emph{large-field polymer partition function} is defined as
	\begin{equation}\label{eq:Is.mag.low4} 
		\Xi_{\Lambda}^{\rm{LT}}(\beta)\;=\;1+\sum_{n\ge 1}\frac{1}{n!}\sum_{(\gamma_1,\ldots,\gamma_n)\in\left[\Gamma^{\Lambda}_{\rm{prim}}\right]^n}\prod_{1\le i<j\le n}^n\zeta(\gamma_i, \gamma_j)\prod_{i=1}^nw^{\rm low}_\beta(\gamma_i)
	\end{equation}
	with
	\begin{equation} \label{eq:Is.mag.low5} 
		w^{\rm low}_\beta(\gamma_i)\;=\;\exp(-2\beta |\gamma_i|), 
	\end{equation}
	and $\zeta$ defined as in \eqref{eq.Ising.model.Xuan.1}.
	
	By cluster expansion theory, we can give a representation of the pressure function with a \(+\)--boundary condition in \(\Lambda\) as the following form
	\begin{equation}\label{eq:Is.mag.low16}
		P^{+}_{\Lambda}(\beta)=\beta\frac{|\mathcal{E}_{\Lambda}|}{|\Lambda|}+\frac{1}{|\Lambda|}\log \Xi^{\rm LT}_{\Lambda}(\beta),
	\end{equation}
	where
	\begin{equation}\label{eq:Is.mag.low6} 
		\log\Xi^{\rm{LT}}_{\Lambda}(\beta)\;=\;\sum_{n=1}^{\infty}\frac{1}{n!}\sum_{(\gamma_1,\ldots,\gamma_n)\in \left[\Gamma^{\Lambda}_{\rm{prim}}\right]^n}a_n^{T}(\gamma_1,\ldots,\gamma_n)\prod_{i=1}^nw^{\rm low}_\beta(\gamma_i),
	\end{equation}
	with Ursell function $a_n^T(\cdot)$ defined as in \eqref{Ursell.F.X1}.
	
	The next theorem establishes a sufficient condition for the existence of the pressure function as $\Lambda \to \mathbb{Z}^d$ in the thermodynamic limit. Based on this condition, we will establish the analytic domain for the pressure function at infinity in the final theorem of this subsection. 
	\begin{Thm}\label{Thm.Xuan.lowtem}  If there exists $a>0$ such that
		\begin{align}\label{eq.Xuan.new.GK.low.1111}
			\sup_{{\rm p}\in \mathbb E_{\mathbb Z^d}}\sum_{\substack{\gamma\in\Gamma^{\mathbb Z^d}_{\rm prim}\\{\rm p} \in \gamma}}w^{\rm low}_\beta(\gamma)\eee^{a|\gamma|}\le \frac{\eee^a-1}{2d},
		\end{align}
		then the following holds:
		
		(i) Denote $\bbd{w}_\beta^{\rm low}:=\{w_\beta^{\rm low}(\gamma)\}_{\gamma\in\Gamma^{\mathbb Z^d}_{\rm prim}}$. For each finite contour $\gamma\in \Gamma_{\rm prim}^{\mathbb Z^d}$, $\card{\Gamma}_{\gamma}(\bbd{w}^{\rm low}_{\beta})$ converges. Furthermore, for finite contour $\gamma\in\Gamma_{\rm prim}^{\mathbb Z^d}$,
		\[\card{\Gamma}_\gamma(\bbd{w}^{\rm low}_{\beta})\le \eee^{a|\gamma|},\]
		where $|\Gamma|_{\gamma}$ is defined in  \eqref{eq.Xuan.Rob13}.
		
		(ii) The free energy function  \eqref{eq:Is.mag.low16}  converges absolutely and uniformly in $\Lambda$, and for a fixed plaquette ${\rm p}\in\mathbb E_{\mathbb Z^d}$, 
		\begin{equation}\label{eq.Xuan.new.12} p^+(\beta, h)=\beta d+d\sum_{B\subset \mathbb Z^d: B\ni {\rm p}}\frac{1}{|B|}\Psi^{\rm low}(B),\end{equation}
		where for each finite set of plaquettes $B\subset \mathbb E_{\mathbb Z^d}$, $\Psi^{\rm low}(\cdot)$ is defined as follows
		\begin{align}
			\Psi^{\rm low}(B)=\sum_{n=1}^{\infty}\frac{1}{n!}\sum_{(\gamma_1\ldots \gamma_n)\in[\Gamma^{\mathbb Z^d}_{\rm prim}]^n\atop \gamma_1\cup \cdots \cup\gamma_n=B}a_n^T(\gamma_1,\ldots,\gamma_n)\prod_{i=1}^nw^{\rm low}_{\beta}(\gamma_i)\nonumber
		\end{align}
	\end{Thm}
	In the last theorem, we aim to specify the domain of temperature $\beta$ for which the pressure function $p^+(\beta)$ is analytic, as stated in the following theorem.
	\begin{Thm}\label{main.result.low.Xuan} The pressure function $p(\beta,h)$ is analytic in the domain $\mathcal{D}$ with
		\[\mathcal{D}\;=\;\{(\beta, h)\in\mathbb R\times \mathbb R:\beta\ge\min_{m,\kappa\in \mathcal{L}}\phi^{\rm low}(m,\kappa)\}\]
		where $\phi^{\rm low}$ is defined by
		\begin{eqnarray}\label{eq.Xuan.low.Ising.new4cor}
			\phi^{\rm low}(m,\kappa)&:=&\log(1+2m)+\left(-2+\frac{3}{d}\right)\log(1+\sqrt{1+(d-1)(\kappa-1)(\kappa+3)m})\nonumber\\ &&-\frac{1}{d}\log(2m)+\frac{2d-2}{d}\log\left (\sqrt{1+(d-1)(\kappa-1)(\kappa+3)m}+2m\kappa+1\right)\nonumber\\ &&-[\log (\kappa-1)-\log(\kappa+3)],
		\end{eqnarray}
		and the domain $\mathcal{L}$ is defined as follows
		\begin{equation}\label{eq.Xuan.low.Ising.domain.new4cor}
			\mathcal{L}:=\left\{(m,\kappa)\in (0,1)\times(0,1): 0<m\le \frac{4(2d-3)\kappa+(\kappa^2+2\kappa-3)(d-1)}{4(2d-3)^2\kappa^2}\right\},
		\end{equation}
		for each $d\ge 2$. 
	\end{Thm}
	The proof of Theorem \ref{Thm.Xuan.lowtem} and \ref{main.result.low.Xuan} are explained in Subsection \ref{subsection.Xuan.low.proof}.
	
	\medskip
	\paragraph{\bf Dimension \( d=2 \)} In the case $d=2$, based on the deformation rule presented in \cite{FV17}, then the compatibility relation is redefined as $\gamma\sim \gamma'$ if and only if $V_\gamma\cap V_{\gamma'}=\varnothing$ where $V_{\gamma}$ is the set of vertices in a primitive contours $\gamma$. Note that $|V_\gamma|=|\gamma|$. Then by an analogous method using to that used Theorem \ref{Thm.Xuan.hightem}, from the Fernandez-Procacci condition yields the  Gruber-Kunz condition as follows:
	\begin{equation}
		\sup_{x\in(\mathbb Z^d)^*}\sum_{x\in V_{\gamma}\atop \gamma\in{\Gamma_{\rm prim}}}w_{\beta}(\gamma)\eee^{a|\gamma|}\;\le\; \eee^a-1.
	\end{equation}
	In the next step, we aim to count the number of primitive contours in \( \mathbb{Z}^2 \), which is equivalent to counting the number of self-avoiding polygons in \( \mathbb{Z}^2 \). In Subsection \ref{subsection.Xuan.high.proof}, we can bound the number of self-avoiding polygons by \( \frac{4}{9} \cdot 3^k \). Additionally, since the number of edges in any self-avoiding polygons must be even,  we obtain
	\begin{eqnarray}
		\sup_{x\in(\mathbb Z^d)^*}\sum_{x\in V_{\gamma}\atop \gamma\in{\Gamma_{\rm prim}}}w_{\beta}(\gamma)\eee^{a|\gamma|}&\le& \frac{4}{9}\sum_{k=2}^{\infty}3^{2k}\eee^{(-4\beta+2a)k}\nonumber\\ &=& \frac{4}{9}\frac{3^4\eee^{-8\beta+4a}}{1-3^2\eee^{-4\beta+2a}}\nonumber\\ &\le& \eee^{a}-1.
	\end{eqnarray}
	By elementary algebra, the temperature $\beta$ below by  
	\begin{equation}
		\beta\ge  \frac{a}{2}-\frac{1}{4}\log\left[\frac{-9(\eee^a-1)+\sqrt{81(\eee^a-1)^2+144(\eee^a-1)}}{72}\right]:= \phi_2^{\rm low}(a).
	\end{equation}
	Then we have
	\begin{equation}
		\beta\ge \min_{a\ge 0}\phi^{\rm low}_2(a)=0.822614.
	\end{equation}
	
	\medskip
	\paragraph{\bf Comparison of analyticity domains.} We will present a comparison of our results with the previous work discussed  by Balister and Bollob\'as \cite{BB07}. According to this reference, analyticity holds if 
	\begin{equation}
		2\beta \ge \frac{1}{2d}+\frac{1}{d}\log(8\eee d^2)+\log\left(1+\frac{1}{d}\right)=\frac{3}{2d}+\frac{1}{d}\log(8d^2)+\log\left(1+\frac{1}{d}\right):=2\beta_{\rm BB}.
	\end{equation}
	for any $d\ge 2$. From the condition \eqref{eq.Xuan.low.Ising.new4}, we consider $\kappa=4d+1$ and $m=1/(16d)$, then $(m, \kappa)\in \mathcal{L}$. More precisely, we have
	\begin{equation}\frac{1}{16d}\le \frac{d(d+1)(d-1)}{16(d-3/2)^2(d+1/4)^2}< \frac{1}{(2d-3)(4d+1)}+\frac{d(d+1)(d-1)}{16(d-3/2)^2(d+1/4)^2},
	\end{equation}
	for $d\ge 2$. Moreover, $\phi^{\rm low}(1/16d,4d+1)$ is equal to
	\begin{align}
		\log\left(1+\frac{1}{8d}\right)&+\frac{1}{d}\log(8d)+\left(-2+\frac{3}{d}\right)\log(1+d)+\left(2-\frac{2}{d}\right)\log\left( d+\frac{3}{2}+\frac{1}{8d}\right)+\log\left(1+\frac{1}{d}\right)\nonumber\\& \le  \log\left(1+\frac{1}{8d}\right)+\frac{1}{d}\log(8d^2)+\left(2-\frac{2}{d}\right)\log\left( 1+\frac{\frac{1}{2d}+\frac{1}{8d^2}}{1+\frac{1}{d}}\right)\nonumber\\&+ \frac{1}{d}\log\left( 1+\frac{1}{d}\right)+\log\left(1+\frac{1}{d}\right)\nonumber\\ &\le\frac{1}{8d}+\frac{1}{d}\log(8d^2) +\left(1-\frac{1}{d}\right) \left(\frac{1}{d}+\frac{1}{4d^2}\right)+ \frac{1}{d^2}+\log\left(1+\frac{1}{d}\right)\nonumber\\ & \le\frac{9}{8d}+\frac{1}{d}\log(8d^2) +\frac{1}{4d^2}+ \log\left(1+\frac{1}{d}\right)\nonumber\\& < \beta_{\rm BB}.
	\end{align}
	To clarify this further, let us  examine the rate function \( r^{\text{high}} \) that measures the ratio between our bound and the Balister--Bollobás bound, defined as follows
	\[r^{\rm low}=\frac{\beta_{NN}}{\beta_{BB}}.\]
	We observe that our bound is stronger than the Balister--Bollobás bound, as illustrated in Figure \ref{Fig.6} and \ref{Fig.7}.
	\begin{figure}[h]
		\centering
		\includegraphics[width=0.5\linewidth]{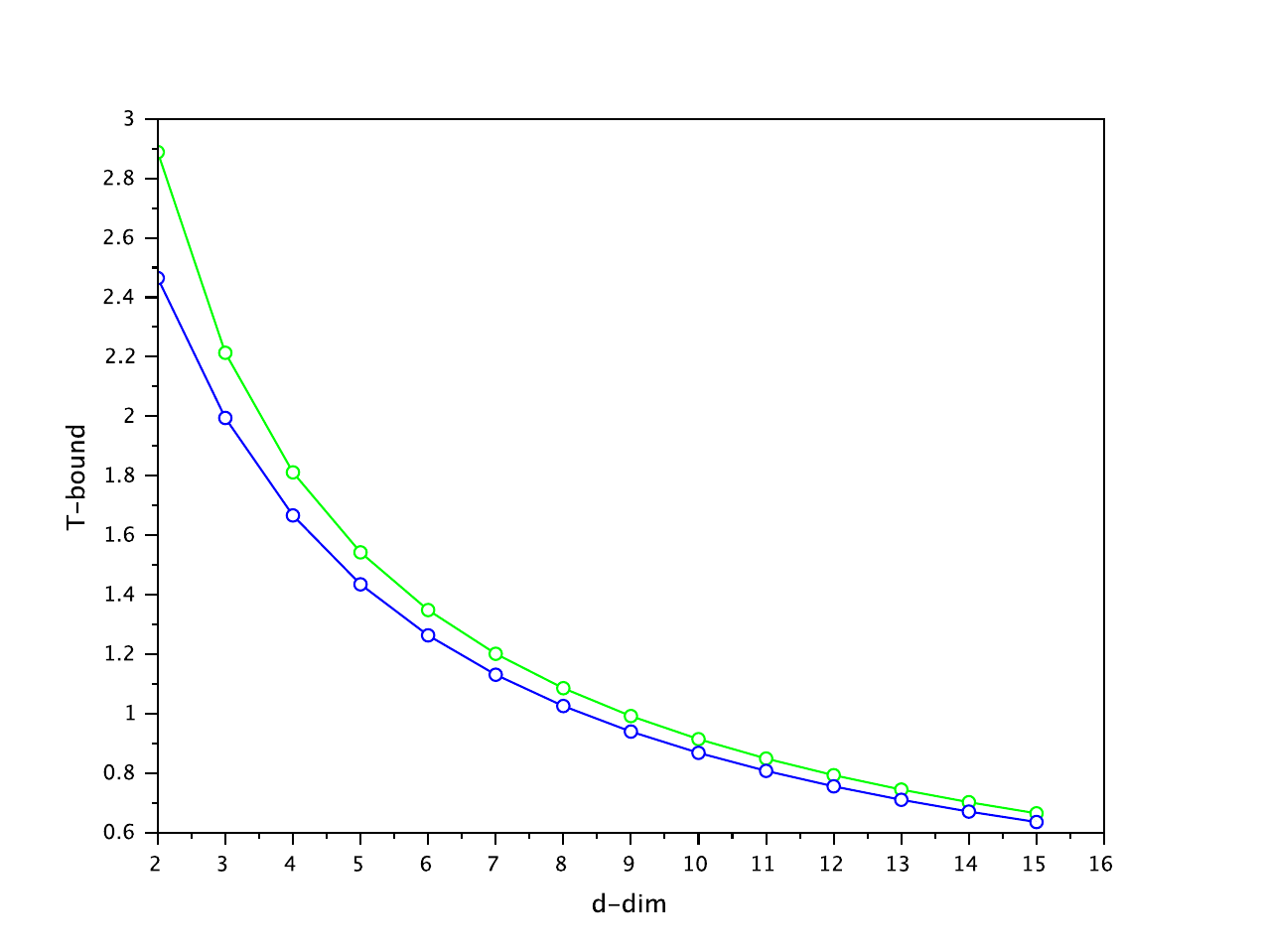}
		\caption{Comparison with the result of Balister and Bollob\'as (the green and red lines represented our result, Balister--Bollob\'as result, respectively)}
		\label{Fig.6}
	\end{figure}
	\begin{figure}[h]
		\centering
		\includegraphics[width=0.5\linewidth]{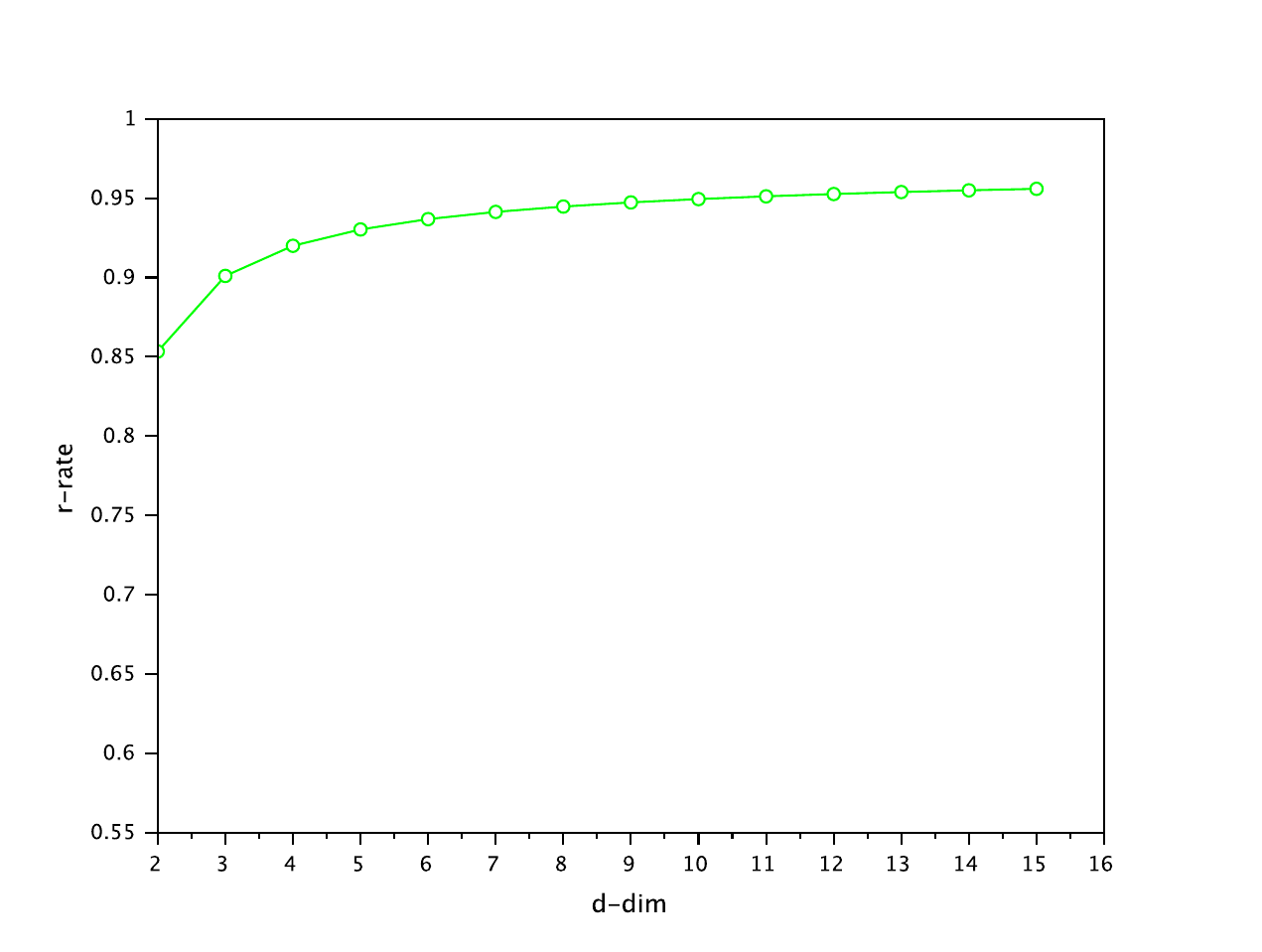}
		\caption{The rate between our bound and Balister and Bollobás bound}
		\label{Fig.7}
	\end{figure}
	
	In particular, for \(d = 2\), Procacci \cite{AP23} obtain an  analyticity domain for 
	\begin{equation}
		\beta\ge 0.94.
	\end{equation}
	In this  case, we obtain a better domain \( \beta \geq 0.822614 \) which is less restrictive than Balister and Bollobás or Procacci.

		\section{PROOFS}\label{sec.proofs}

		\subsection{Proof of Theorem \ref{Thm.Xuan.strong.1} and \ref{main.result.X1}}\label{Sub.Xuan.new.1}
		In this subsection, we will utilize the Fernandez-Procacchi criterion given in Section \ref{sec.app.clu2} along with a new compatible relation presented in Section \ref{sec.Xuan.1} to derive the improved Gruber-Kunz condition which is presented in the  following proof.
		\begin{proof}[{\bf Proof of Theorem \ref{Thm.Xuan.strong.1}}] 
			\noindent
			
			(i)
			Let us start with the following direct consequence of from Proposition \ref{thm:FP} that $\card{\Gamma}_S(\bbd{w}_{\beta,h})$ converges if for each $S\in\mathcal{P}$, 
			\begin{equation}\label{eq:Is.mag8}
				1+\sum_{n\ge 1}\sum_{(S_1,...,S_n)\in \mathcal{P}^n\atop S\nsim S_i,\; S_i\sim S_j,\; 1\le i, j\le n}\prod_{i=1}^nw_{\beta,h}(S_i)\eee^{a(S_i)}\;\le\; \eee^{a(S)}
			\end{equation}
			where we take $\mu_S=w_{\beta,h}(S)\eee^{a(S)}$. It is easy to see that
			\[ S\nsim S^\prime \Longleftrightarrow d(S,S')\le 1\Longleftrightarrow S\cap [S^\prime]_1=\varnothing,\] with
			\[[S]_1\;:=\;\{j\in \mathbb{Z}^d:d(j,S)\le 1\},\]
			and 
			\[ S\sim S^\prime \Rightarrow d(S,S')> 1\Rightarrow S\cap S'=\varnothing.\] 
			
			It implies that the left-hand side of convergence  condition \eqref{eq:Is.mag8} can be bounded as follows
			\begin{equation}\label{eq:Is.mag9}
				1+\sum_{n\ge 1}\sum_{\{S_1,...,S_n\}\subset \mathcal{P}\atop [S_0]_1\cap S_i\ne \varnothing,\; S_i\cap S_j=\varnothing,\; 1\le i, j\le n}\prod_{i=1}^nw_{\beta,h}(S_i)\eee^{a(S_i)}.
			\end{equation}
			Then we can replace the convergent condition \ref{eq:Is.mag8} by
			\begin{align}1+\sum_{n\ge 1}\sum_{\{S_1,...,S_n\}\subset \mathcal{P}\atop [S]_1\cap S_i\ne \varnothing,\; S_i\cap S_j=\varnothing,\; 1\le i, j\le n}&\prod_{i=1}^nw_{\beta,h}(S_i)\eee^{a|[S_i]_1|}\le \eee^{a|[S]_1|}
			\end{align}
			By an argument analogous to that used to derive the Gruber-Kunz condition \eqref{eq:int-ccct10}, the constraint, $[S]_1\cap S_i\ne\varnothing,\;S_i\cap S_j=\varnothing,\; 1\le i<j\le n$,  implies that each of the  polymers $S_1,\ldots,S_n$ must intersect different points in $[S]_1$ to avoid overlapping. Consequently, we can conclude that: (i) $n\le \card{[S]_1}$, and (ii) there are $n$ different points in $[S]_1$ touched by $S_1\cup\ldots\cup S_n$. The selection of these points can be done in $\binom{\card{[S]_1}}{n}$ ways.  Hence the left-hand side of \eqref{eq:int-ccct8} is less than or equal to
			\begin{equation}\label{eq:Is.mag10}
				1+\sup_{x\in\mathbb Z^d}\sum_{x\in S\atop S\in\mathcal{P}}w_{\beta,h}(S)\eee^{a|[S]_1|}\;\le\; \eee^a.
			\end{equation}
			
			(ii)
			We can then rearrange the terms of the cluster expansion in finite subset $X\subset \Lambda$ as follows:
			\begin{align}\label{eq:int-ccct.X.5}
				\sum_{X: X \subset \Lambda} \Psi(X) 
				&= \sum_{s \in \Lambda} \sum_{\substack{X:\\ s \in X \subset \Lambda}} \frac{1}{|X|} \Psi(X) \\
				&= \sum_{s \in \Lambda} \left\{ \sum_{\substack{X\subset  \mathbb{Z}^d:\\ s \in X}} \frac{1}{|X|} \Psi(X) - \sum_{\substack{X\subset  \mathbb{Z}^d:\\ s \in X, X \not\subset \Lambda}} \frac{1}{|X|} \Psi(X) \right\}. 
			\end{align}
			The difference between the two series is well-defined, since both series are absolutely convergent as a consequence of Proposition \ref{thm:FP}. Notice that both of them contain clusters of unbounded sizes. In the case Ising model on lattice $\mathbb Z^d$, we have one useful property called translation invariance. By translation invariance, the first sum over $X$ in the right-hand side of \eqref{eq:int-ccct.X.5} does not depend on $s$ and thus yields 
			\begin{equation}
				\sum_{s \in \Lambda}  \sum_{\substack{X\subset \mathbb Z^d:\\ s \in X}} \frac{1}{|X|} \Psi(X)=|S|\sum_{\substack{X\subset \mathbb V:\\ s \in X}} \frac{1}{|X|} \Psi(X),
			\end{equation}
			for a fixed $s\in \Lambda$. The second sum in the right-hand side of \eqref{eq:int-ccct.X.5} is a boundary term. Indeed, whenever $s \in X \not\subset \Lambda$, there must exist at least one component $S_k \in \{S_1,\ldots S_n\}$ which intersects the boundary of $\Lambda$, namely $S_k \cap \partial \Lambda \neq \emptyset$. 
			Therefore, using \eqref{eq.FP.1} for the second inequality,
			\begin{equation}
				\left| \sum_{s \in \Lambda} \sum_{\substack{X:\\ s \in X, X \not\subset   \Lambda}} \frac{1}{|X|} \Psi(X) \right|
				\leq |\partial \Lambda| \max_{j \in \mathbb Z^d} \sum_{X : X \ni j} |\Psi(X)| \leq \eee^a|\partial \Lambda|.
			\end{equation}
			We thus obtain
			\begin{equation}
				\frac{1}{|\Lambda|}\left|\log Z(\bbd{\rho})- \sum_{s \in S}  \sum_{\substack{X\subset  \Lambda\\ s \in X}} \frac{1}{|X|} \Psi(X)\right|\le \eee^a\frac{|\partial \Lambda|}{|\Lambda|}.
			\end{equation}
			Therefore under the condition \eqref{eq.Xuan.new.GK.1111} and taking the thermodynamic limit  along a sequence of boxes $\Lambda$ such that $|\partial \Lambda|/|\Lambda|\to 0$, the boundary term vanishes, leaving $|\mathcal{E}_{ \Lambda}| / |\Lambda| \to d$ and yielding
			\begin{equation}
				\psi_\beta(h) = \beta d + h + \sum_{X: X \ni 0} \frac{1}{|X|} \Psi(X).
			\end{equation}
		\end{proof}
		Before proceeding with further calculations, let us establish a weaker condition for the convergence of the power series \(|\Gamma|_{S}(\bbd{w}(\beta, h))\) in the following lemma. This condition arises from a bound on the weight \(w_{\beta,h}(\cdot) \eee^{a|[S]_1|}\), as outlined in this lemma, along with the Gruber-Kunz condition. This bound is particularly useful for estimating the parameters \(\beta\) and $h$.
		
		\begin{lem}\label{lem.Xuan.new.1} 
			If there exists $a>0$ such that
			\begin{align}\label{eq.Xuan.new.GK.1}
				\sum_{k=1}^{\infty}\card{\mathcal{A}_k}\eee^{-2d\beta\;k^{(d-1)/d}V(1)^{1/d}-2h\,k+(2d+1)ak}\le \eee^a-1
			\end{align}
			with
			\begin{equation}\label{eq.Xuan.new.11}
				\mathcal{A}_k\;:=\;\{S\in\mathcal P:\;0\in S,\;\card{S}\;=\;k\},
			\end{equation}
			then the condition \eqref{eq.Xuan.new.GK.1111} holds.
		\end{lem}
		\begin{proof}
			We observe that
			\begin{equation}\label{eq:Is.mag10.1}
				\card S\;\le \;\card{[S]_1}\;\le \;(2d+1)\card S,
			\end{equation}
			Then we have
			\begin{equation}\label{eq:Xuan.Is.mag10.new.1}
				\sup_{x\in\mathbb Z^d}\sum_{x\in S\atop S\in\mathcal{P}}w_{\beta,h}(S)\eee^{a|[S]_1|}\;\le\; \sup_{x\in\mathbb Z^d}\sum_{x\in S\atop S\in\mathcal{P}}\eee^{-2\beta |\partial S|-2h|S|}\eee^{(2d+1)a|S|}
			\end{equation}
			We also observe that the smallest ratio of area to volume is achieved by a $d$-dimensional sphere.  Denoting the volume and surface area of a sphere of radius $R$, respectively, by
			\[\begin{split}V_d(R)&=V_d(1)\,R^d\\S_d(R)&=dV_d(1)\,R^{d-1}\end{split}\]
			we obtain 
			\begin{align}\label{eq:Xuan.Is.mag10.new.2}
				\card{\partial S}\ge S_d(R)=dV_d(1)^{1/d}[V_d(R)]^{(d-1)/d} &=dV_d(1)^{1/d}|S|^{(d-1)/d}.
			\end{align}
			As a consequence of inequalities \eqref{eq:Xuan.Is.mag10.new.1}, \eqref{eq:Xuan.Is.mag10.new.2}, and the condition \eqref{eq.Xuan.new.GK.1}, we obtain
			\begin{align}\label{eq:Is.mag11}
				\sup_{x\in\mathbb Z^d}\sum_{x\in S\atop S\in\mathcal{P}}w_{\beta,h}(S)\eee^{a|[S]_1|}&\le \sum_{k=1}^{\infty}\card{\mathcal{A}_k}\eee^{-2d\beta\;k^{(d-1)/d}V(1)^{1/d}-2h\,k+(2d+1)ak}\nonumber\\ &\le \eee^a-1
			\end{align}
			with $\mathcal{A}_k$ defined as in \eqref{eq.Xuan.new.11}. This concludes the proof. 
		\end{proof}
		For simplicity in computation, we will bound 
		\[2d\beta\;k^{(d-1)/d}V(1)^{1/d} \ge 0.\]
		This implies that we can replace condition \eqref{eq.Xuan.new.GK.1} with a weaker condition as follows:
		\begin{equation}\label{eq.Xuan.new.GK1.1}
			\sum_{k=1}^{\infty}\card{\mathcal{A}_k}\eee^{-2h\,k+(2d+1)ak}\le (\eee^a-1)
		\end{equation}
		To estimate \( h > 0 \) such that the condition \eqref{eq.Xuan.new.GK1.1} is satisfied, we will utilize the generating function method. This method was previously employed by Balister and Bollobás in \cite{BB07} to establish a bound on the number of connected subsets. Let
		\begin{equation}
			p(x)=\sum_{n=1}b_nX^n
		\end{equation}
		be a generating  function of a number of connected subsets, and
		\begin{equation}
			p(\eee^{-2h+(2d+1)a})=\sum_{k=1}^{\infty}\card{\mathcal{A}_k}\eee^{-2h\,k+(2d+1)ak}.
		\end{equation}
		\begin{proof}[\bf Proof of Theorem \ref{main.result.X1}] Each connected subset \( S \) with volume \( n \) corresponds to a connected graph \( G_S \) based on that set. For any pair of elements \( s, s' \in S \), the set \( \{s, s'\} \) is considered an edge in the connected graph \( G_S \) if the distance \( d(s, s') = 1 \). Let \( L \) denote the maximum distance from any vertex in \( S \) to the root vertex \( 0 \). We define the generating function as
			\begin{equation}
				p_L(X):=\sum_{n\ge 1}b_{n, L}X^n
			\end{equation}
			where $b_{n, L}$ counts the number of connected subsets where the distance between the root \( 0 \) and each site is less than or equal to \( L \). Let us see the relation between the $p_{L+1}(X)$ and $p_L(X)$. We know that each connected subset \( G_S \), where the distance from the root \( 0 \) to each site is less than or equal to \( L+1 \), can be represented by connecting the root \( 0 \) with smaller connected subsets whose maximum distance to their respective roots is less than or equal to \( L \). These smaller subsets are rooted at some or all of the adjacent sites of \( 0 \). Then
			\begin{equation}
				p_{L+1}(X) = X (1 + p_L(X))^{2d},
			\end{equation}
			and $p_0(X)=X$. As $L$ increases, $b_{n,L}$ increases, and for $L \geq n$, $b_{n,L}$ is constant, say $b_{n,L} = a_n$. Thus $p_L(X)$ increases monotonically to 
			\[
			p(X) = \sum_{n\ge 1} b_n X^n,
			\]
			provided that $X$ lies within the radius of convergence of this limiting series, 
			where $p(X)$ satisfies the equation
			\begin{equation}\label{eq.Xuan.Ising.strong.SSS1}
				p(X) = X (1 + p(X))^{2d}
			\end{equation}
			We then rewrite equation \eqref{eq.Xuan.Ising.strong.SSS1} as
			\begin{equation}
				X = p (1 + p)^{-2d}
			\end{equation}
			and maximize $X$. If the maximum value \( X = X_c \) occurs at \( p = p_c \), we can see inductively that \( p_L(X_c) \leq p_c \) for all \( L \). Therefore, the generating function \( p(X) \) converges for all \( X \leq X_c \). To estimate the maximum of \( X \), we will solve the following equation
			\begin{equation}
				\frac{dX}{dp} = (1+p)^{-2d}\left(1-\frac{2dp}{p+1}\right)=0
			\end{equation}
			Thus, \( X \) reaches its maximum value when \( p = \frac{1}{2d - 1} \), which leads us to the equation
			\begin{equation}
				X_c=\left(1+\frac{1}{2d-1}\right)^{-(2d-1)}(2d)^{-1}.
			\end{equation}
			In order to get the optimal domain for the analyticity of pressure function we will set up 
			\begin{equation}
				p(X_1)=(\eee^a-1),
			\end{equation}
			and 
			\[X_1=(\eee^a-1)
			\left[1 + (\eee^a-1)
			\right]^{-2d}=(\eee^a-1)\eee^{-2da}.\]
			Since $p(X)$ is an increasing function, it implies that for $X \leq X_1$, we have $p(X) \leq p(X_1)$. Therefore, we have
			\begin{equation}
				X=\eee^{-2h+(2d+1)a}\le X_1. 
			\end{equation}
			It is equivalent to
			\begin{eqnarray}\label{eq.min.strong.X.3}
				2h
				&\ge& (2d+1)a+2d\log\left[1+(\eee^a-1)
				\right]-\log(\eee^a-1)\nonumber\\&=&(4d+1)a-\log(\eee^a-1)=\phi^{\rm st}(a).
			\end{eqnarray}
			Therefore, we have
			\begin{equation}
				2h
				\ge \min_{a> 0}\phi^{\rm st}(a).
			\end{equation}
			Using elementary calculus, $\phi^{\rm st}$ attains the minimum value at 
			\begin{equation}
				a^{\rm st}=\log\left[1+\frac{1}{4d}\right],
			\end{equation}
			and 
			\begin{equation}\label{eq.min.strong.X.1}
				\min_{a> 0}\phi^{\rm st}(a)=(4d+1)\log\left[1+\frac{1}{4d}\right]+\log(4d)=\varphi^{\rm st}(d).
			\end{equation}
		This completes the proof of Theorem \ref{main.result.X1}.
		\end{proof}

		\subsection{Proof of Theorem  \ref{Thm.Xuan.hightem} and \ref{main.result.high.X1}}\label{subsection.Xuan.high.proof}
		At the beginning of this subsection, we will start with the Fernandez-Procacchi criterion presented in Section \ref{sec.app.clu2}, along with a new compatible relation discussed in Section \ref{sec.Xuan.1}. This will allow us to derive an improvement of the Gruber-Kunz condition, which is presented in the following proof.
		\begin{proof}[\bf Proof of Theorem  \ref{Thm.Xuan.hightem}]
			\noindent 
			
			(i) Let us start with following readily from Proposition  \ref{thm:FP} that $\card{\Gamma}_{\mathfrak{p}}(\bbd{w}^{\rm high}_{\beta})$ converges if for each $\mathfrak{p}\in\mathcal{P}^{\rm p}_{\mathbb Z^d}$, 
			\begin{equation}\label{eq:Is.mag8}
				1+\sum_{n\ge 1}\sum_{(\mathfrak{p}_1,...,\mathfrak{p}_n)\in [\mathcal{P}^{\rm p}_{\mathbb Z^d}]^n\atop \mathfrak{p}\nsim \mathfrak{p}_i,\; \mathfrak{p}_i\sim \mathfrak{p}_j,\; 1\le i, j\le n}\prod_{i=1}^nw^{\rm high}_{\beta}(\mathfrak{p}_i)(\coth a)^{|\mathfrak{p}_i|}\;\le\; (\coth a)^{|\mathfrak{p}|}
			\end{equation}
			where we take $\xi_{\mathfrak{p}}=w^{\rm high}_{\beta}(\mathfrak{p})(\coth a)^{|\mathfrak{p}|}$. It is easy to see that
			\[ \mathfrak{p}\nsim \mathfrak{p}^\prime \Longleftrightarrow \mathfrak{p},\mathfrak{p}'\mbox{ have at least one common edge (denoted by } \mathfrak{p}\cap\mathfrak{p}'\ne\varnothing) \mbox{ for }  \;\mathfrak{p},\mathfrak{p}'\in\mathcal{P}^{\rm p}_{\mathbb Z^d}.\]
			
			We then rewrite the convergent condition \eqref{eq:Is.mag8} by
			\begin{align}\label{eq:int-ccct10.high1}
				1+\sum_{n\ge 1}\sum_{\{\mathfrak{p}_1,...,\mathfrak{p}_n\}\subset \mathcal{P}^{\rm p}_{\mathbb Z^d}\atop \mathfrak{p}\cap \mathfrak{p}_i\ne \varnothing,\; \mathfrak{p}_i\cap\mathfrak{p}_j=\varnothing,\; 1\le i, j\le n}\prod_{i=1}^nw^{\rm high}_{\beta}(\mathfrak{p}_i)(\coth a)^{|\mathfrak{p}_i|}\le (\coth a)^{|\mathfrak{p}|}
			\end{align}
			By using an argument similar to the one employed to derive the Gruber-Kunz condition \eqref{eq:int-ccct10}, we start with the constraint in the sum: $\mathfrak{p}_1 \cap \mathfrak{p}_i \neq \varnothing$ and $\mathfrak{p}_i \cap \mathfrak{p}_j = \varnothing$ for $1 \le i < j \le n$. This implies that each of the polymers $\mathfrak{p}_1, \ldots, \mathfrak{p}_n$ must intersect different edges in $\mathfrak{p}$ to avoid overlapping. Consequently, we have the following results: (i) \( n \leq \card{\mathfrak{p}} \), and (ii) there are \( n \) distinct points in \( \mathfrak{p} \) that are touched by \( \mathfrak{p}_1 \cup \ldots \cup \mathfrak{p}_n \). The selection of these points can be accomplished in \( \binom{\card{\mathfrak{p}}}{n} \) ways. Therefore, we can conclude that the left-hand side of \eqref{eq:int-ccct10.high1} is less than or equal to 
			\begin{equation}\label{eq:Is.mag.high.10}
				1+\sup_{e\in\mathcal{E}_{\mathbb{Z}^{d}}}\sum_{e\in \mathfrak{p}\atop \mathfrak{p}\in\mathcal{P}^{\rm p}_{\mathbb Z^d}}w^{\rm high}_{\beta}(S)(\coth a)^{|\mathfrak{p}|}\;\le\; \coth a,
			\end{equation}
			Then we obtain
			\begin{equation}\label{eq:Is.mag.high.11}
				\sup_{x\in\mathbb Z^d}\sum_{x\in \gamma\atop \gamma\in\mathcal{P}^{\rm p}_{\mathbb Z^d}}w^{\rm high}_{\beta}(\gamma)(\coth a)^{|\mathfrak{p}|}
				\;\le\; \frac{\coth a-1}{2}.
			\end{equation}
			(ii) Statement (ii) can be proved using an  argument analogous to that of  Theorem \ref{Thm.Xuan.strong.1} (ii). This completes the proof.
		\end{proof}
		Denote
		\[\mathcal{N}_k=\big\{\mathfrak{p}\in \mathcal{P}^{\rm p}_{\mathbb Z^d}\big|e\in\mathfrak{p},\;|\mathfrak{p}|=k\big\},\]
		for a fixed edge $e\in\mathcal{E}_{\mathbb Z^d}$.
		\begin{lem}\label{lem.Eulerian.cycle} For each \(\gamma \in \mathcal{C}^{\rm cy}\), the number of edges in \(\gamma\) is an even value.
		\end{lem}
		\begin{proof}
			Since $\mathbb Z^d$ is hypercube graph then $\mathbb Z^d$ is bipartite graph. And we know that a bipartite graph is a graph that does not contain any odd-length cycles. Then each \(\gamma \in \mathcal{C}^{\rm cy}\), the number of edges in \(\gamma\) is an even value. The proof is completed.
		\end{proof}
		In the next step, we will establish an upper bound for the set \(\mathcal{N}_k\). This set consists of self-avoiding polygons, a topic that has been extensively studied, as summarized in \cite{AG09}. In the following lemma, we will present a well-known upper bound for \(\mathcal{N}_k\), which is derived from the properties of simple cycles. The proof is straightforward and can be found in \cite{AG09}, so we will omit it here.
		\begin{lem}\label{lem.number.cycle} For $k\ge 4$, 
			\begin{equation}\label{eq:Is.mag.high10} 
				\card{\mathcal{N}_k}\;\le\; 2(2d-1)^{k-1}.
			\end{equation}
		\end{lem}
		
		\noindent
		\begin{proof}[\bf{Proof of Theorem \ref{main.result.high.X1}:}] As the consequence of the alternative Gruber-Kunz condition, Lemma \ref{lem.Eulerian.cycle}, Lemma \ref{lem.number.cycle}, we obtain
			\begin{equation}\label{eq:Is.mag.high11}
				\frac{2}{2d-1}\sum_{k\ge 2}\bigg((2d-1)\frac{\tanh\beta}{\tanh a}\bigg)^{2k}\;=\;\frac{2}{2d-1}\frac{\left[(2d-1)\frac{\tanh\beta}{\tanh a}\right]^4}{1-\left[(2d-1)\frac{\tanh\beta}{\tanh a}\right]^2}\;\le\; \coth a-1.
			\end{equation}
			It is also equivalent to the following expression 
			\begin{equation}\label{eq:Is.mag.high13}\left[(2d-1)\frac{\tanh\beta}{\tanh a}\right]^4+ (2d-1)\frac{\coth a-1}{2}\left[(2d-1)\frac{\tanh\beta}{\tanh a}\right]^2-(2d-1)\frac{\coth a-1}{2}\le 0.
			\end{equation}
			We then have
			\begin{align}\label{eq:rr.mag.high11}\tanh\beta&\le \displaystyle\frac{\tanh a}{\sqrt{2d-1}}\sqrt{\frac{(1-\coth a)+\sqrt{(\coth a-1)^2+\frac{8}{2d-1}(\coth a-1)}}{4}}\nonumber\\ &=\phi^{\rm high}(a).
			\end{align}
			In the next step, we would like to optimize the best domain for temperature by taking
			\begin{align}\label{eq:rr.mag.highX11}\tanh\beta&\le \displaystyle\max_{a> 0}\frac{\tanh a}{\sqrt{2d-1}}\sqrt{\frac{(1-\coth a)+\sqrt{(\coth a-1)^2+\frac{8}{2d-1}(\coth a-1)}}{4}}.
			\end{align}
			Using basic optimization theory, we can show that there is a unique number \(\bar{a}^{\text{high}} \in (\mathrm{arcoth}(1+1/3),\infty)\) at which \(\phi^{\text{high}}\) reaches its global maximum when \(a > 0\). Consequently, the global maximum of \(\phi^{\text{high}}\) is given by \(\phi^{\text{high}}(\bar{a}^{\text{high}})\). Therefore, the condition \eqref{eq:rr.mag.highX11} implies that the function is analytic for 
			\begin{equation}\label{eq:Is.mag.high15}
				\beta\;\le \;\frac{1}{2}\log \frac{1+\phi^{\rm high}(\bar{a}^{\rm high})}{1-\phi^{\rm high}(\bar{a}^{\rm high})}\;=\; \beta_N.
			\end{equation}
			This completes the proof of Theorem \ref{main.result.high.X1}.
		\end{proof}
		
		\subsection{Proof of Theorem \ref{Thm.Xuan.lowtem}  and \ref{main.result.low.Xuan}}\label{subsection.Xuan.low.proof} 
		We will begin by examining a consequence of the Fernandez-Procacci convergence condition to derive the statement in Theorem \ref{Thm.Xuan.lowtem}.
		\begin{proof}[\bf Proof of Theorem \ref{Thm.Xuan.lowtem}]
			\noindent 
			
			(i.) It follows readily from the Fern\'andez-Procacci condition of Proposition   \ref{thm:FP} guarantees the convergence of expansion if
			\begin{equation}\label{eq:Is.mag.low7} 
				1+\sum_{n\ge 1}\sum_{\{\gamma_1,...,\gamma_n\}\subset\Gamma^{\mathbb Z^d}_{\rm prim}}\prod_{i=1}^nw_\beta(\gamma_i)\eee^{a(\gamma_i)}\prod_{i=1}^n\one{\gamma_0\nsim \gamma_i}\prod_{1\le i<j\le n}\one{\gamma_i\sim \gamma_j}\;\le\; \eee^{a(\gamma_0)}
			\end{equation}
			To apply the Fern\'andez-Procacci criterion, as for most of the models, we set $\mu_{\gamma}=w({\gamma})\eee^{a\card{\gamma}}$ to obtain
			\begin{equation}\label{eq:int-ccct8}
				1+\sum_{n\ge 1}\sum_{\{\gamma_1,\ldots,\gamma_n\}\subset \Gamma^{\mathbb Z^d}_{\rm prim}\atop \gamma_0\nsim \gamma_i,\;\gamma_i\sim \gamma_j,\; 1\le i<j\le n}\prod_{i=1}^{n}w_\beta(\gamma_i)\eee^{a|\gamma_i|}\;\le\;\eee^{a|\gamma_0|}
			\end{equation}
			From the constraints in the sum of inequality \eqref{eq:int-ccct8}, for each $\{\gamma_1,\ldots,\gamma_n\}$,  we can assume, without loss of generality, we have \( k \) contours \( \gamma_{i_1}, \ldots, \gamma_{i_k} \) that share common plaquettes with \( \gamma_0 \), and \( (n-k) \) contours \( \gamma_{i_{k+1}}, \ldots, \gamma_{i_n} \) that do not share common plaquettes with \( \gamma_0 \), however \( V_{\gamma_{i_j}} \cap V_{\gamma_0} \neq \varnothing \). Note that for each pair \( \gamma_i \) and \( \gamma_j \) where \( 1 \leq i < j \leq n \), there are no common plaquettes, and \( V_{\gamma_i} \cap V_{\gamma_j} = \varnothing \).
			Denote 
			\[\gamma^{0,i}=\{\mathrm{p}: {\rm p} \,\mbox{ is common plaquette of contour }\,\gamma_0 \mbox{ and contour  }\gamma_{i}\}.\] 
			For $j=k+1,\ldots, n$, since $\gamma^{0,i_j}=\varnothing$ and $V_{\gamma_{i_j}}\cap V_{\gamma_0}\ne \varnothing$, $\gamma_0$ intersects $\gamma_{i_j}$ at a $(d-2)$-dimensional faces,  ${\rm f}^{i_j}$. Let $x^{i_j}\in V_{\gamma_{i_j}}\cap V_{\gamma_0}$ be the site  adjacent to one of  $(d-2)$-dimensional faces ${\rm f}^{i_j}_1$. Let ${\rm p}^{i_j}\in \gamma_0$ be the plaquette that passes through the $(d-2)$-dimensional face ${\rm f}^{i_j}$ and adjacent to the site $x^{i_j}$. We need to show that ${\rm p}^{i_j}\notin \gamma^{0,1}\cup \cdots\cup\gamma^{0,k}$. Let us assume that there exists $1\le\ell\le k$ such that ${\rm p}^{i_j}\in \gamma^{0.\ell}$. Then ${\rm p}^{i_j}$ is the common plaquettes of $\gamma_0$ and $\gamma_{i_\ell}$. Since $\gamma_{i_j}$ and $\gamma_{i_\ell}$ do not share any  plaquettes, they $\gamma_{i_\ell}$ and $\gamma_{i_j}$ must intersect at at least $(d-2)$-dimensional face $f^{i_j}$. It implies that $V_{\gamma_{i_j}}\cap V_{\gamma_{i_\ell}}\ne \varnothing$. This contradicts the compatibility.
			Moreover, for $k+1\le\ell<\kappa\le  n$, $\gamma_{\ell}\sim\gamma_{\kappa}$ then $V_{\gamma_{\ell}}\cap V_{\gamma_\kappa}=\varnothing$, it implies that $(V_{\gamma_\kappa}\cap V_{\gamma_0})\cap (V_{\gamma_\ell}\cap V_{\gamma_0})=\varnothing$. 
			
			The explanation of the constraint in the previous paragraph implies that we have $k$-contours $\gamma_{i_1},\ldots, \gamma_{i_k}$ which must intersect different plaquettes ${\rm p}_1,\ldots, {\rm p}_k\in\gamma$ to avoid overlapping, and $(n-k)$-contours $\gamma_{i_{k+1}},\ldots, \gamma_{i_n}$ such that $V_{\gamma_{i_{k+1}}},\ldots, V_{\gamma_{i_n}}$ must intersect different sites $x^{i_{k+1}},\ldots, x^{i_n}\in V_{\gamma}$. Furthermore, 
			${\rm p}_{i_j}\in\gamma_{0}\setminus \{{\rm p}_1,\ldots, {\rm p}_k\}$ are the adjacent plaquettes of $x^{i_j}$ for $j=k+1,\ldots, n$. Consequently, we can conclude that: $n\le \card{\gamma_0}$ such that  there exist  $k$ distinct plaquettes ${\rm p}_1,\ldots, {\rm p}_k\in \gamma$ which are also the plaquettes in $\gamma_{i_1}\cup\ldots\cup \gamma_{i_k}$ and $n-k$ distinct plaquettes in $\gamma_0\setminus \{{\rm p}_1,\ldots, {\rm p}_k\}$. These $n-k$ plaquettes are adjacent of $(n-k)$ different sites in $V_{\gamma_0}\setminus \cup_{j=1}^kV_{\gamma_{i_j}}$ touched by $V_{\gamma_{i_{k+1}}}\cup\cdots\cup V_{\gamma_{i_n}}$. The selection of these points can be done in $\displaystyle\binom{\card{\gamma_0}}{k}\binom{\card{\gamma_0}-k}{n-k}$ ways.  
			Hence the left-hand side of \eqref{eq:int-ccct8} is less than or equal to
			\begin{align}\label{eq:int-ccct9}
				&1+\sum_{n=1}^{\card{\gamma_0}}\sum_{k=0}^n\binom{\card{\gamma_0}}{k}\binom{\card{\gamma_0}-k}{n-k}\Bigg[\sup_{{\rm p}\in \gamma_0}\sum_{\substack{\gamma\in\Gamma^{\mathbb Z^d}_{\rm prim}\\ \gamma\ni {\rm p}}}w_\beta({\gamma})\eee^{a|\gamma|}\Bigg]^{k}\Bigg[\sup_{{\rm p}\in \gamma_0}\sum_{\substack{\gamma\in\Gamma^{\mathbb Z^d}_{\rm prim}\\ V_\gamma\ni  x_{\rm p}, {\rm p\notin \gamma}}}w_\beta(\gamma)\eee^{a|\gamma|}\Bigg]^{n-k}\nonumber\\
				& =\;\Bigg[1+\sup_{{\rm{p}}\in \gamma_0}\sum_{\substack{\gamma\in\Gamma^{\mathbb Z^d}_{\rm prim}\\ \gamma\ni {\rm p}}}w_{\beta}({\gamma})\eee^{a|\gamma|}+\sup_{{\rm p}\in \gamma_0}\sum_{\substack{\gamma\in\Gamma^{\mathbb Z^d}_{\rm prim}\\ V_\gamma\ni  x_{\rm p}, {\rm p\notin \gamma}}}w_\beta(\gamma)\eee^{a|\gamma|}\Bigg]^{\card{\gamma_0}}
				\nonumber\\ &\le \Bigg[1+\sup_{{\rm{p}}\in \mathbb{E}_{\mathbb Z^d}}\sum_{\substack{\gamma\in\Gamma^{\mathbb Z^d}_{\rm prim}\\ \gamma\ni {\rm p}}}w_\beta({\gamma})\eee^{a|\gamma|}+\sup_{{\rm p}\in \mathbb E_{\mathbb Z^d}}\sum_{\substack{\gamma\in\Gamma^{\mathbb Z^d}_{\rm prim}\\ V_\gamma\ni  x_{\rm p}, {\rm p\notin \gamma}}}w_\beta(\gamma)\eee^{a|\gamma|}\Bigg]^{\card{\gamma_0}}
			\end{align}
			where the site $x_{\rm p}\in V_{\gamma_0}$ is adjacent to the plaquette ${\rm p}$. 
			This leads us to the following sufficient condition for \eqref{eq:int-ccct8}
			\begin{equation}\label{eq:int-ccct.lowX10b}
				\sup_{{\rm{p}}\in \mathbb{E}_{\mathbb Z_*^d}}\sum_{\substack{\gamma\in\Gamma^{\mathbb Z^d}_{\rm prim}\\ \gamma\ni {\rm p}}}w_\beta({\gamma})\eee^{a|\gamma|}+\sup_{{\rm p}\in \mathbb E_{\mathbb Z^d}}\sum_{\substack{\gamma\in\Gamma^{\mathbb Z^d}_{\rm prim}\\ V_\gamma\ni  x_{\rm p}, {\rm p\notin \gamma}}}w_\beta(\gamma)\eee^{a|\gamma|}\;\le \;\eee^{a}-1.
			\end{equation}
			We have 
			\begin{equation}
				\sup_{{\rm p}\in \mathbb E_{\mathbb Z^d}}\sum_{\substack{\gamma\in\Gamma^{\mathbb Z^d}_{\rm prim}\\ V_\gamma\ni  x_{\rm p}, {\rm p\notin \gamma}}}w_\beta(\gamma)\eee^{a|\gamma|}\le (2d-1)\sup_{{\rm p}\in \mathbb E_{\mathbb Z^d}}\sum_{\substack{\gamma\in\Gamma^{\mathbb Z^d}_{\rm prim}\\{\rm p} \in \gamma}}w_\beta(\gamma)\eee^{a|\gamma|}
			\end{equation}
			This brings us to the following sufficient condition for  the condition \eqref{eq:int-ccct.lowX10b}
			\begin{equation}\label{eq:int-ccct.lowX10}
				\sup_{{\rm p}\in \mathbb E_{\mathbb Z^d}}\sum_{\substack{\gamma\in\Gamma^{\mathbb Z^d}_{\rm prim}\\{\rm p} \in \gamma}}w_\beta(\gamma)\eee^{a|\gamma|}\le \frac{\eee^a-1}{2d}.
			\end{equation}
			
			(ii) We can prove statement (ii) by using an argument similar to the one used for  Theorem \ref{Thm.Xuan.strong.1}  (ii).
		\end{proof}
		To verify the inequality \eqref{eq:int-ccct.lowX10}, we need to count the number of simple closed surfaces \(\gamma\) in \(\mathbb{Z}^d\) that pass through a specific plaquette \({\rm p}\). 
		In our work, we build upon the results of Balister and Bollob\'as in \cite{BB07}, where they employ the generating function method to provide a bound on the number of contours on the lattice \(\mathbb{Z}^d\). In this paper, we present several updates that utilize optimal tools to enhance our results compared to those of Balister and Bollob\'as in \cite{BB07}.
		
		Following the notation in \cite{LM98}, let $\gamma_i$ represent the set of plaquettes of a contour $\gamma$ that are orthogonal to the coordinate axis number $i$, where $i = 1, \ldots, d$. We define $i^*$ as the direction for which $\card{\gamma_{i^*}} = \min_{i} \card{\gamma_i}$. This direction, $i^*$, is referred to as the $\gamma$-vertical direction. Consequently, all the plaquettes of $\gamma$ are categorized into two groups: horizontal plaquettes, belonging to $\gamma_{i^*}$, and vertical plaquettes, which are those in $\gamma \setminus \gamma_{i^*}$. Based on these definitions, we can construct a floor-stack multi-graph that pertains to primitive contours, as outlined in \cite{BB07}.
		
		Let  $h$ be the  generating function for the number of primitive contours containing the plaquette ${\rm p}$ defined as
		\begin{equation}
			h(x)=\sum_{n}a_n X^n
		\end{equation}
		where $a_n$ stands for the possible number of primitive contours with volume $n$. A generating function $g$ is then  defined as follows
		\[ g(X, Y) = \sum_{r,s} a_{r,s} X^r Y^s, \]
		and  satisfies the following equation
		\begin{equation}g(X,Y)=X(1+\kappa g(X,Y))^{2d-2}
		\end{equation}
		where 
		\begin{equation}\label{eq.Xuan.ising.low.new.1}
			\kappa=1+\frac{4Y}{1-Y},
		\end{equation}
		and \( a_{r,s} \) represents a bound on the number of possible spanning trees with a total stack size of \( s \) and a total floor volume of \( r \).  
		
		Let us define 
		\[\widetilde g(X, Y)=\frac{4Y}{(1-Y)^2}[g(X,Y)]^2=\sum_{n=1}^{\infty}nY^{n}[2g(X, Y)]^2.\]
		In the following proposition, we will provide the upper bound of the generating function \( h \) for the number of primitive contours. This result provide a significantly contribution to estimating the domain of analyticity of the pressure function.
		\begin{prop}\label{prop.Xuan.lowX.new2} For all $d\ge 2$,
			\begin{equation}\label{BB.X.low.1}
				h(X^{1/d}Y)\le g(X,Y)+(d-1)\widetilde g(X,Y)
			\end{equation}
			for every $0<Y<1$, and $X\le X_c$ defined as
			\[X_c=\frac{1}{(2d-3)\kappa}\left[\frac{2d-3}{2d-2}\right]^{2d-2}.\]
		\end{prop}
		\begin{proof} Let us rewrite the generating function $h$ as follows 
			\begin{equation}\label{eq.X.low.new1}
				h(x)=h^=(x)+h^\bot(x).
			\end{equation}
			In expression \eqref{eq.X.low.new1}, \( h^{=}(x) \) represents the generating function associated with the root \( {\rm p} \) from the set of floors, while \( h^{\bot}(x) \) denotes the generating function related to the root \( {\rm p} \) from the set of verticals. To establish a bound for the generating function \( h \), we will estimate each term, \( h^{=}(x) \) and \( h^{\bot}(x) \).
			\medskip
			
			\noindent
			{\bf Case 1. ${\rm p}$ belongs to the set of floors.} Let \( B = \partial B \) be a primitive contour. According to Lemma 4 in \cite{BB07}, the floor-stack graph \( G \) of the primitive contour \( B \) is connected. Fix a spanning tree of \( G \). We can reconstruct the floors by specifying each floor as a rooted \((d-1)\)-complex along with the connecting stacks. In this step, we will prove 
			\begin{equation}\label{eq.Xuan.Ising.low.1}
				h^{=}(X)\le g(X,Y)
			\end{equation}
			for each $X\le X_c$ and $0<Y<1$. To establish the inequality \eqref{eq.Xuan.Ising.low.1}, we first prove by induction on $L$ that 
			\begin{equation}\label{eq.Xuan.Ising.low.2}
				h_L^{=}(X)\le g_L(X,Y)
			\end{equation}
			where the graph length $L$ is defined as
			\begin{equation}
				L:=\max\{{\rm dist}_G(u,{\rm p}): u\in \mathbb V_G\}
			\end{equation}
			with ${\rm p}$ chosen as the root vertex here, ${\rm dist}_G(u,v)$ denotes the minimum number of edges connecting vertices  $u$ and $v$ in the graph $G$. The functions  $g_{L}(X, Y)$ are defined recursively by: $g_0(X,Y)=0$, and
			\[g_{L+1}(X,Y)=X[1+\kappa g_{L}(X,Y)]^{2d-2}\]
			where $\kappa$ is given by \eqref{BB.X.low.1}. Let \( h_L^{=}(x) \) be the generating function associated with the primitive contour \( G \) that passes through the root \( {\rm p} \). The distance from the root \( {\rm p} \) to the set of floors in a spanning tree graph \( G \) is less than or equal to \( L \). We have \( h_0^{=}(X) = 0 = g_0(X, Y) \). We assume that this holds for \( L \), meaning \( h_0^{=}(X) = 0 = g_0(X, Y) \). We need to prove that the inequality stated in \eqref{eq.Xuan.Ising.low.2} also holds for \( L + 1 \).
			
			Consider any primitive contour $B$ whose associated graph $G$ has a the spanning tree with length  at most $L+1$ relative to the root ${\rm p}$.  We examine the growth of this graph starting from the root \({\rm p}\). For each of the \(2(d-1)\) faces of \({\rm p}\), we have three options:
			\begin{itemize}
				\item[a.)] Attach nothing.
				\item[b.)] Attach a neighboring horizontal \((d-1)\)-cube, which extends the current floor. In this case, when we attach cubes, we will continue building the complex from the new horizontal \((d-1)\)-cube.
				\item[c.)] Attach a stack that includes a horizontal \((d-1)\)-cube at the other end. When choosing to attach a stack, it can extend in one of two directions (up or down), and the horizontal \((d-1)\)-cube at the other end of the stack can be attached in one of two positions. Since the stack itself can be of any positive integral length, this choice yields
				\[4(Y + Y^2 + Y^3 + \ldots)= \frac{4Y}{1-Y}.\]
			\end{itemize}
			Additionally, note that two stacks cannot be  connected sequentially, as this would simply merge them into a single stack. Likewise, we cannot attach a stack to a horizontal cube because doing so would push the stack beyond the boundaries of the current floor. From options a.), b.), c.) and the induction hypotheses that $h_L^{=}(x)\le g_{L}(X, Y)$, we can bound 
			\begin{equation}
				h^{=}_{L+1}(X)\le X\left[1+\kappa h_L(X)\right]^{2d-2}\le X\left[1+\kappa g_L(X, Y)\right]^{2d-2}:=g_{L+1}(X, Y).
			\end{equation}
			It is easy to prove that $g_L$ is an increasing function with respect to the variable $X \in (0,1)$ for each $Y \in (0,1)$. Let us consider the following equation
			\begin{equation}
				X=g(1+\kappa g)^{-2d+2}
			\end{equation}
			We can compute that $X$ attains its maximum value at 
			\begin{equation}
				g_c=\frac{1}{(2d-3)\kappa},
			\end{equation}
			and the maximum value of $X$ is equal to
			\begin{equation}
				X_c=g_c[1+\kappa g_c]^{-2d+2}=\frac{1}{(2d-3)\kappa}\left[\frac{2d-3}{2d-2}\right]^{2d-2}.
			\end{equation}
			We can  prove by induction that $g_{L}(X, Y)\le g_c$, for each $Y\in (0,1)$ and $X\le X_c$. More precisely, $g_0(X,Y)=0\le g_c$. Assume that $g_L(X, Y)\le g_c$ holds. We need to prove that $g_{L+1}(X, Y)\le g_c$ also holds. We have
			\[g_{L+1}(X, Y)=X(1+\kappa g_L(X, Y))^{2d-2}\le X_c(1+\kappa g_c)^{2d-2}=g_c.\] 
			Thus, for every $Y \in (0,1)$, $g_{L}(X, Y)$ converges monotonically to $g(X, Y)$, where $g(X,Y)$ is the solution to the following equation
			\begin{equation}\label{eq.Xuan.Ising.low.3}
				g(X,Y)=X(1+\kappa g(X,Y))^{2d-2}.
			\end{equation}
			{\bf Case 2. ${\rm p}$ belongs to the set of verticals.}
			We begin constructing the spanning tree of graph \( G \), starting from the root \( {\rm p} \), which is situated in the middle of a stack. There are floors at each end of the stack. Then \( h^{\bot}(x) \) is bounded from above by
			\begin{equation}
				(Y+2Y^2+3Y^3+\cdots)[2g(X, Y)]^2=\frac{4Y}{(1-Y)^2}[g(X, Y)]^2.
			\end{equation}
			The term $kY^k (2g(X, Y))^2$ arises from selecting a stack of length $k$, where there are $k$ possible choices for the root. The contour then grows from this root, beginning with two floors, each initiating in one of two possible directions.
			
			We fix a primitive contour $B$ and initially choose the vertical direction to be $i$. The contour $B$ then contributes a term $X^r Y^s$ to the generating function $g(X, Y)$, where $r = |B|^{=}$ counts the number of horizontal edges, and $s \leq |B|$ denotes the total stack height in the vertical direction $i$. According to Lemma 4 in \cite{BB07}, the contour $B$ contributes at least once for every spanning tree of the associated floor-stack graph.
			
			Rather than fixing a single vertical direction $i$, we can consider each of the $d$ directions as vertical, one at a time. In this generalized setting, the contour $B$ contributes at least
			\[
			\sum_{i=1}^d X^{r_i} Y^{s_i}
			\]
			to the function $g(X, Y) + \widetilde{g}(X, Y)$, where each $r_i = |B^=|$ and $s_i \leq |B|$ are measured with respect to the chosen vertical direction $i$. Note that the root ${\rm p}$ is vertical in $d-1$ dimensions and horizontal in only one. Consequently, since $r_i$ represents the number of horizontal components in direction $i$, the total horizontal contribution satisfies $n = \sum_{i=1}^d r_i$, and each vertical stack height satisfies $s_i \leq n$. The AM-GM inequality, along with the condition $0< Y<1$ leads to the conclusion that
			\begin{equation}
				\sum_{i=1}^d X^{r_i} Y^{s_i}\ge dX^{n/d}Y^{\sum s_i/d}\ge dX^{n/d}Y^{n}.
			\end{equation}
			Consequently, for any $0 < Y < 1$ and $0 < X < X_c= X_c(Y)$, we have
			\begin{equation}
				g(X, Y) + \widetilde{g}(X, Y)\ge dh(X^{1/d}Y).
			\end{equation}
			This completes the  proof.
		\end{proof}
		\begin{proof}[\bf{Proof of Theorem \ref{main.result.low.Xuan}}] As the consequence of Theorem \ref{Thm.Xuan.lowtem}, we need to find a number $\beta$ such that
			\begin{equation}
				d\sup_{{\rm p}\in \mathbb E_{\mathbb Z^d}}\sum_{\substack{\gamma\in\mathcal{P}_{\mathbb V}\\{\rm p} \in \gamma}}w_\beta(\gamma)\eee^{a|\gamma|}=d\sum_{n\ge 1}a_n\eee^{(-2\beta+2a)n}=dh(\eee^{-2\beta+2a})\le \frac{\eee^a-1}{2}:=m.
			\end{equation}
			In order to do it, let us begin with following identity
			\begin{eqnarray}g(X, Y)+(d-1)\widetilde g(X, Y)&=&g(X, Y)+(d-1)\frac{4Y}{(1-Y)^2} g^2(X, Y)\nonumber\\  &=& g(X, Y)+(d-1)\frac{(\kappa-1)(\kappa+3)}{4} g^2(X, Y)=m,
			\end{eqnarray}
			where $\kappa$ is defined in \eqref{eq.Xuan.ising.low.new.1}. Since $g\ge 0$, we have
			\begin{equation}g_1=2\frac{-1+\sqrt{1+(d-1)(\kappa-1)(\kappa+3)m}}{(d-1)(\kappa-1)(\kappa+3)}=\frac{2m}{1+\sqrt{1+(d-1)(\kappa-1)(\kappa+3)m}}.
			\end{equation}
			From equation \eqref{eq.Xuan.Ising.low.3}, we have
			\begin{eqnarray}
				X_1&=&g_1(1+\kappa g_1)^{-2d+2}\nonumber\\ &=&\frac{2m}{1+\sqrt{1+(d-1)(\kappa-1)(\kappa+3)m}}\left(1+\frac{2m\kappa}{1+\sqrt{1+(d-1)(\kappa-1)(\kappa+3)m}}\right)^{-2d+2}.
			\end{eqnarray}
			Note that if 
			\begin{equation}\label{eq.Xuan.Ising.low.new3}
				m\le \frac{4(2d-3)\kappa+(\kappa^2+2\kappa-3)(d-1)}{4(2d-3)^2\kappa^2}
			\end{equation}
			we have $g_1 \le g_c$. Since $g(1 - \kappa g)^{-2d + 2}$ is an increasing function on $[0, g_c]$ and $g_1 \le g_c$, it follows that $X_1 \le X_c$. By Proposition \ref{prop.Xuan.lowX.new2}, we have 
			\begin{equation}
				dh(X_1^{1/d}Y)\le g(X_1, Y)+(d-1)\widetilde g(X_1,Y)= g_1+(d-1)\frac{(\kappa-1)(\kappa+3)}{4}g_1^2=m
			\end{equation}
			for each $Y\in (0,1)$. Since \( h \) is an increasing function, it implies that for every \( X \le (X_1)^{1/d} Y \), we have \( h(X) \le m \). Therefore, we conclude that
			\begin{equation}
				\eee^{-2\beta+a}\le (X_1)^{1/d}Y=(X_1)^{1/d}\frac{\kappa-1}{\kappa+3}.
			\end{equation}
			It is equivalent to stating that for each $m$ and  $\kappa$ satisfying the condition \eqref{eq.Xuan.Ising.low.new3}, we obtain
			\begin{eqnarray}\label{eq.Xuan.low.Ising.new4}
				2\beta&\ge& a-\frac{1}{d}\log X_1-\log (\kappa-1)+\log(\kappa+3)\nonumber\\ &=&\log(1+2m)-\frac{1}{d}\log X_1-[\log (\kappa-1)-\log(\kappa+3)]\nonumber\\ &=&\log(1+2m)-\frac{1}{d}\log\frac{2m}{1+\sqrt{1+(d-1)(\kappa-1)(\kappa+3)m}}\nonumber\\ && +\frac{2d-2}{d}\log\left[1+\frac{2m\kappa}{1+\sqrt{1+(d-1)(\kappa-1)(\kappa+3)m}}\right]-[\log (\kappa-1)-\log(\kappa+3)]\nonumber\\ &=&\log(1+2m)+\left(-2+\frac{3}{d}\right)\log(1+\sqrt{1+(d-1)(\kappa-1)(\kappa+3)m})-\frac{1}{d}\log(2m)\nonumber\\ &&+\frac{2d-2}{d}\log\left (\sqrt{1+(d-1)(\kappa-1)(\kappa+3)m}+2m\kappa+1\right)-[\log (\kappa-1)-\log(\kappa+3)]\nonumber\\ &&=\phi^{\rm low}(m,\kappa).
			\end{eqnarray}
			To achieve a better domain, we can optimize the parameters \( m \) and \( \kappa \) as follows
			\begin{equation}
				2\beta\ge \min_{m,\kappa\in \mathcal{L}}\phi^{\rm low}(m,\kappa).
			\end{equation}
			where $\mathcal{L}$ is defined as in \eqref{eq.Xuan.low.Ising.domain.new4cor},
			for each $d\ge 2$. This completes the proof.
		\end{proof}
		\section*{Acknowledgements}
		The authors wishes to express his deepest gratitude to Prof. Roberto Fern\'andez for his helpful suggestions during the preparation of this  paper.  
		This research is funded by Vietnam Ministry of Education and Training (MOET) under grant number B2026-CTT-03.
		
		\appendix
		
		\section{Cluster expansion for subset gases}\label{sec.app.clu2}
		Subset gases are specific types of polymer gases that are frequently utilized in cluster expansion within statistical mechanics. Their definition requires a countable subset, denoted by $\mathbb{V}$ (e.g. $\mathbb{Z}^d,\;\mathcal{E}_{\mathbb{Z}^d},\;\mathbb{E}_{\mathbb{Z}^d}$) which acts as an underlying ``space''. Polymers are defined as finite, non-empty subsets of  $\mathbb{V}$, represented mathematically as
		\[\mathcal{P}_{\mathbb V}\;=\;\{S\subset\mathbb{V}:0\;<\;\card{S}\;<\;\infty\}.\]
		with a compatibility relation, denoted by \(S\sim S'\), which is a subset of
		\(\mathcal{P}_{\mathbb V}\times\mathcal{P}_{\mathbb V}\) satisfying
		\(S\sim S\) for every \(S\in\mathcal{P}_{\mathbb V}\).
		The compatibility relation depends on the model under consideration. For instance, in Section \ref{sec.Xuan.1}, we stated that $S\sim S'$ if and only if $d(S, S')\ge 2$. In the work of Bissacot, Procacci and Fernandez \cite{BFP10}, it mentioned that $S\sim S'$ if and only if $S\cap S'=\varnothing$. 
		Polymers can be measured through its cardinality, so it makes sense to talk about large and small polymers. The definition of the gas is completed by a family of activities $\bbd{z}\;=\;\{z_S\in\mathbb{C}\}_{S\in\mathcal{P}_{\mathbb V}}$. 
		Let us define the partition function for gas polymers as follows
		\begin{align}
			Z(\bbd{z})=1+\sum_{n\ge 1}\frac{1}{n!}\sum_{(S_1,\ldots, S_n)\in \mathcal{P}_{\mathbb V}^n}\prod_{1\le i<j\le n}\zeta(S_i, S_j)\prod_{i=1}^nz_{S_i}\nonumber 
		\end{align}
		Using Mayer's trick (which can be found in \cite{MMay40}), we can derive $\log Z(\bbd{z})$ as the following form:
		\begin{align}\label{eq.Xuan.Rob12}\log Z(\boldsymbol{z})=\sum_{n\ge 1}\frac{1}{n!}\sum_{(S_1,\ldots, S_n)\in\mathcal{P}_{\mathbb V}^n}a^T_n(S_1,\ldots,S_n)\prod_{i=1}^nz_{S_i},
		\end{align}
		where $a_n^T(\cdot)$ is defined in \eqref{Ursell.F.X1}.
		
		To study the convergence of cluster expansion, we typically examine it through the convergence conditions of the formal power series in infinite volume as below (see \cite{FP07} for a full explanation): For each $S\in\mathcal{P}_{\mathbb V}$,
		\begin{align}\label{eq.Xuan.Rob13}|\Gamma|_S(\bbd{\rho})=1+\sum_{n=1}^{\infty}\frac{1}{n!}\sum_{(S_1,\ldots, S_n)\in\mathcal{P}_{\mathbb V}^n}|a^T_{n+1}(S_1,\ldots,S_n)|\prod_{i=1}^{n}\rho_{S_i}
		\end{align}
		with $\bbd{\rho}\in [0,\infty)^{\mathcal{P}_{\mathbb V}}$.
		
		In the following proposition, we will give the condition for the convergence of the cluster expansion presented in equation \eqref{eq.Xuan.Rob13}. This condition is well-known as the Fernández-Procacci criterion, first introduced in \cite{FP07}.
		\begin{prop} [Fern\'andez-Procacci criterion]\label{thm:FP} Suppose that for some $\bbd{\xi}\in [0,\infty)^{\mathcal{P}_{\mathbb V}}$ there exists $\bbd{\mu}\in [0,\infty)^{\mathcal{P}_{\mathbb V}}$ such that 
			\begin{equation}\label{eq:int-ccct4}
				\xi_{S_0}\psi^{\mathrm{FP}}_{S_0}(\bbd{\mu})\;\le\;\mu_{S_0}, \quad\text{for each}\;S_0\in\mathcal{P}_{\mathbb V}
			\end{equation}
			with
			\begin{equation}\label{eq:int-ccct5}
				\psi^{\mathrm{FP}}_{S_0}(\bbd{\mu})=1+\sum_{n\ge 1}\frac{1}{n!}\sum_{ (S_1,\ldots,S_n)\in\mathcal{P}_{\mathbb V}^n\atop S_0\nsim S_i,\; S_i\sim S_j,\; 1\le i, j\le n}\prod_{i=1}^n\mu_{S_i}.
			\end{equation}
			Then  $\card{\Gamma}_{S}(\bbd{\xi})$ is  convergent. Furthermore,  for each $S\in\mathcal{P}_{\mathbb V}$,
			\begin{equation}\label{eq.FP.1}
				\xi_{S}\card{\Gamma}_{S}(\bbd{\xi})\le \mu_{S}.
			\end{equation}
		\end{prop}
	The proof can be found in full detail in the reference \cite{FP07}. Moreover, the statement (i) in the  Theorem \ref{Thm.Xuan.strong.1}, Theorem \ref{Thm.Xuan.hightem}, and Theorem \ref{Thm.Xuan.lowtem} are the consequences of Proposition \ref{thm:FP} which is presented more precisely in Section \ref{sec.proofs}.

		To apply the Fern\'andez-Procacci criterion, as for  most of the models,  we set $\mu_{\gamma}\;=\;\xi_{\gamma}\eee^{a\card{\gamma}}$ to obtain
		\begin{equation}\label{eq:int-ccct8}
			1+\sum_{n\ge 1}\sum_{\{S_1,\ldots,S_n\}\subset\mathcal{P}_{\mathbb V}\atop S_0\cap S_i\ne\varnothing,\;S_i\cap S_j=\varnothing,\; 1\le i<j\le n}\prod_{i=1}^{n}\xi_i\eee^{a|S_i|}\;\le\;\eee^{a|S_0|}
		\end{equation}
		From the constraint in the sum, $S_0\cap S_i\ne\varnothing,\;S_i\cap S_j=\varnothing,\; 1\le i<j\le n$,  this means that each of the  polymers $S_1,\ldots,S_n$ must intersect different points in $S_0$ to avoiding overlapping. Consequently, we can conclude that: (i) $n\le \card{S_0}$, and (ii) there are $n$ different points in $S_0$ touched by $S_1\cup\ldots\cup S_n$. The selection of these points can be done in $\binom{\card{S_0}}{n}$ ways. Hence the left-hand side of \eqref{eq:int-ccct8} is less than or equal to
		\begin{align}\label{eq:int-ccct9}
			1+\sum_{n=1}^{\card{\gamma }}\binom{\card{\gamma}}{n}\Bigg[\sup_{x\in \gamma}\sum_{\substack{S\in\mathcal{P}_{\mathbb V}\\ S\ni x}}\xi_{S}\eee^{a(S)}\Bigg]^{n} =\;\Bigg[1+\sup_{x\in S_0}\sum_{\substack{S\in\mathcal{P}_{\mathbb V}\\ S\ni x}}\xi_{S}\eee^{a(S)}\Bigg]^{\card{S_0}}
		\end{align}
		This leads us to the following sufficient condition for \eqref{eq:int-ccct8}
		\begin{equation}\label{eq:int-ccct10}
			\sup_{x\in S_0}\sum_{\substack{S\in\mathcal{P}_{\mathbb V}\\ S\ni x}}\xi_{S}\eee^{a|S|}\;\le \;\eee^{a}-1.
		\end{equation}
This condition is the well-known Gruber--Kunz condition, originally introduced in \cite{GK71}.

		
		\appendix

		\renewcommand{\refname}{REFERENCES}
		\makeatletter
		\renewcommand\@biblabel[1]{#1.}
		\makeatother

	\end{document}